\documentclass[a4paper,10pt,notitlepage]{article}
\usepackage[dvips]{graphicx}
\usepackage{amsfonts,amssymb}
\usepackage{amsmath,amsthm}
\usepackage{bbm}
\usepackage{mathrsfs}
\usepackage{makeidx}
\usepackage{xypic}
\usepackage{pstricks}
\usepackage{pst-node}
\usepackage{stmaryrd}
\usepackage{url}
\usepackage{youngtab}
\usepackage{cite}
\usepackage{oldgerm}
\usepackage[arrow, matrix, curve]{xy}

\xymatrixcolsep{0.5cm}
\xymatrixrowsep{0.5cm}

\numberwithin{equation}{section}

\theoremstyle{definition}

\newtheorem{lemma}{Lemma}[section]
\newtheorem{cor}{Corollary}[section]

\newcommand{\beqa}{\begin{eqnarray}}
\newcommand{\eeqa}{\end{eqnarray}}
\newcommand{\beq}{\begin{equation}}
\newcommand{\eeq}{\end{equation}}

\newcommand{\calS}{\mathcal{S}}

\newcommand{\calL}{\mathcal{L}}
\newcommand{\calO}{\mathcal{O}}
\newcommand{\tG}{\textsf{G}}\newcommand{\tB}{\textsf{B}}

\newcommand{\gl}[2]{\ensuremath{\text{gl}\left({#1}|{#2}\right)}}
\newcommand{\un}[2]{\ensuremath{\text{u}\left({#1}|{#2}\right)}}
\newcommand{\pgl}[1]{\ensuremath{\text{pgl}\left({#1}|{#1}\right)}}
\newcommand{\sgl}[2]{\ensuremath{\text{sl}\left({#1}|{#2}\right)}}
\newcommand{\psl}[1]{\ensuremath{\text{psl}\left({#1}|{#1}\right)}}
\newcommand{\psu}[1]{\ensuremath{\text{psu}\left({#1}|{#1}\right)}}
\newcommand{\osp}[2]{\ensuremath{\text{osp}\left({#1}|{#2}\right)}}
\newcommand{\GL}[2]{\ensuremath{\text{GL}\left(\text{{#1}}|\text{{#2}}\right)}}
\newcommand{\U}[2]{\ensuremath{\text{U}\left(\text{{#1}}|\text{{#2}}\right)}}

\newcommand{\PSL}[1]{\ensuremath{\text{PSL}\left(\text{{#1}}|\text{{#1}}\right)}
}
\newcommand{\PSU}[1]{\ensuremath{\text{PSU}\left(\text{{#1}}|\text{{#1}}\right)}
}
\newcommand{\OSP}[2]{\ensuremath{\text{OSP}\left(\text{{#1}}|\text{{#2}}\right)}
}
\newcommand{\bgl}[1]{\ensuremath{\text{gl}\left(\text{{#1}}\right)}}

\newcommand{\bsgl}[1]{\ensuremath{\text{sl}\left(\text{{#1}}\right)}}
\newcommand{\bso}[1]{\ensuremath{\text{so}\left(\text{{#1}}\right)}}
\newcommand{\bsp}[1]{\ensuremath{\text{sp}\left(\text{{#1}}\right)}}
\newcommand{\BGL}[1]{\ensuremath{\text{GL}\left(\text{{#1}}\right)}}

\newcommand{\BSO}[1]{\ensuremath{\text{SO}\left(\text{{#1}}\right)}}
\newcommand{\BSP}[1]{\ensuremath{\text{SP}\left(\text{{#1}}\right)}}
\newcommand{\CP}[1]{\ensuremath{\text{$\mathbb{CP}$}^{\text{{#1}}-1|\text{{#1}}}
}}

\newcommand\str{\text{str}}
\newcommand\sdim{\text{sdim}\,}
\newcommand\id{\mathbbm{1}}

\newcommand{\vac}[1]{\ensuremath{\left< \, #1\, \right>}}

\DeclareMathOperator{\Ker}{\textsf{Ker}}
\DeclareMathOperator{\Img}{\textsf{Im}}
\DeclareMathOperator{\Coh}{\textsf{H}}
\DeclareMathOperator{\Ad}{Ad}
\DeclareMathOperator{\rank}{rank}
\DeclareMathOperator{\topmod}{top}

\title{\textbf{Cohomological Reduction of Sigma Models}}

\author{Constantin Candu$^1$, Thomas Creutzig$^2$,\\ Vladimir Mitev$^1$,
Volker Schomerus$^1$\bigskip \\$^1$DESY
Hamburg, Theory Group, \\
Notkestrasse 85, D--22607 Hamburg, Germany\\$^2$Department of Physics and
Astronomy\\
University of North Carolina, Phillips Hall, CB 3255\\
Chapel Hill, NC 27599-3255, USA
\bigskip\\ E-mail: \text{Constantin.Candu@desy.de},
\text{creutzig@physics.unc.edu},\\ \text{Vladimir.Mitev@desy.de},
\text{Volker.Schomerus@desy.de}}
\date{}

\begin{document}
\begin{titlepage}
    \maketitle
    \begin{abstract}
This article studies some features of quantum field theories with
internal supersymmetry, focusing mainly on 2-dimensional
non-linear sigma models
which take values in a coset superspace.
It is discussed how BRST operators from the target space
supersymmetry algebra can be used to identify subsectors which are
often simpler than the original model and may allow for an
explicit computation of correlation functions. After an extensive
discussion of the general reduction scheme, we present a number of
interesting examples, including symmetric superspaces
$G/G^{\mathbb{Z}_2}$ and coset superspaces of the form
$G/G^{\mathbb{Z}_4}$.

\end{abstract}
\vspace*{-13cm} {\tt {DESY 10-003}}
%
\\
\end{titlepage}

\tableofcontents

\addtolength{\baselineskip}{3pt}
\section{Introduction}

Due to their appearance in many quite distinct areas of physics,
non-linear sigma models with target space (i.e.\ internal)
supersymmetry have been the subject of much interest lately. One
area in which they appear are the proposed dualities between
string theories in Anti de Sitter backgrounds $AdS_{n+1}\times M$
and conformal gauge theories, the most well known examples being
$AdS_5\times S^5$ and $AdS_4 \times \mathbb{CP}^3$ which are
described e.g.\ in \cite{Metsaev:1998it,Berkovits:2000fe,
Stefanski:2008ik, Fre:2008qc,Arutyunov:2009ga, Arutyunov:2008if}. Other lines of
applications involve dense polymers in two dimensions
\cite{Read:2007ti, Candu:2009pj}, the quantum Hall plateau
transitions \cite{Zirnbauer:1999ua} or disordered electron systems
\cite{Wiegmann:1988qn}.

Sigma models on target superspaces possess a number of surprising
properties which are gradually being uncovered. In particular,
there exists several basic series of models which give rise to
families of conformal field theories with continuously varying
exponents, including the supergroup manifolds PSL(N$|$N),
OSP(2N+2$|$2N) and a number of quotients thereof
\cite{Berkovits:1999im,Bershadsky:1999hk,Zirnbauer:1999ua,Read:2001pz,
Babichenko:2006uc}. Note that quantum conformal symmetry does not
require the addition of any Wess-Zumino term, in contrast to the case of
purely bosonic target spaces.

Solving conformal field theories with continuously varying
exponents requires developing entirely new techniques which go far
beyond the conventional algebraic methods. Numerical and algebraic
studies of lattice discretizations \cite{Read:2001pz,Candu:2008vw,
Candu:2008yw} and supersymmetry aided all-order perturbative
computations of spectra \cite{Mitev:2008yt,Candu:2009ep} have been
applied with astonishing results. In some cases is was possible to
determine exact formulas for all (boundary) conformal weights as a
function of the continuous couplings (moduli) of the models.

Having gained some control over the weights it is a natural next
step to investigate properties of higher correlation functions.
While general correlators seem way out of reach, we will be able
to gain useful insights into correlation functions involving a
special subset of fields. Some inspiration can be taken from the
study of conformal field theories with $N=(2,2)$ world-sheet
supersymmetry. For such models, a very conventional trick that one
exploits through the so-called topological twists, is to identify
special subsectors whose dependence on the couplings can be
brought under control. The idea is to employ a fermionic
world-sheet symmetry generator as a BRST operator and to select
its cohomology as the relevant subsector. If the action of the
model is trivial in cohomology the correlation functions of
subsector operators do not depend on the coupling constants of the
theory. Such correlators can then be calculated in the (classical)
limit, as described for example in \cite{Labastida:1997pb}.

The models we are interested in possess target space rather than
world-sheet supersymmetry. A natural idea then is to promote an
internal nilpotent symmetry to a BRST operator. In following this
lead, we shall uncover a rather remarkable structure. Suppose we
are starting with a sigma model on the quotient $G/G'$, defined
as the set of right $G'$ cosets in $G$, with
$G'$ being some sub-supergroup  of the supergroup $G$. Let $Q$ be
some fermionic generator in the superalgebra $\mathfrak{g}'
\subset \mathfrak{g}$ such that $Q^2 =0$. Note that such a $Q$ is
a symmetry of the $G/G'$ sigma model. Through its cohomology, $Q$
defines a subsector. Quite remarkably, the latter turns out to
form the state space of another sigma model on the coset superspace $H/H'$ with a
new pair of supergroups $H' \subset H$. The target space $H/H'$
has smaller dimension than $G/G'$ and the symmetry algebra
$\mathfrak{h}$ of the reduced theory is contained in the symmetry
algebra $\mathfrak{g}$. In many cases, further reduction is
possible until the procedure terminates because the remaining
symmetry algebra does not contain any further nilpotent generators.
Thereby, we obtain a chain of models
$\left\{\mathcal{M}_{\alpha}\right\}_{\alpha\in A}$ which is
parametrized by elements $\alpha$ of some partially ordered set
$A$. The model $\mathcal{M}_\alpha$ is a subsector of
$\mathcal{M}_\beta$, i.e.\  $\mathcal{M}_{\alpha}\subset
\mathcal{M}_{\beta}$, whenever $\alpha < \beta$. Let us give just
one example here. It is provided by the following family of symmetric
superspaces
\begin{equation} \label{eq:ex1M}
\mathcal{M}^{{\text{U/U$^2$}}}_{(\alpha_1,\alpha_2)}(R,S)  \ \cong
\ \frac{\U{R+$\alpha_1$}{$\alpha_1$}}{\U{S+$\alpha_2$}{$\alpha_2$}
\times \U{$\alpha_1$-$\alpha_2$}{R-S+$\alpha_1$-$\alpha_2$}} \ \
\end{equation}
where $R,S, R-S$ and $\alpha_1,\alpha_2, \alpha_1-\alpha_2$ are
all taken to be non-negative integers. The family \eqref{eq:ex1M}
includes the complex projective spaces $\mathbb{CP}^{\text{R}+
\alpha_1-1|\alpha_1}$ for $S =1$ and $\alpha_2 = 0$.

In order to select conformal quotients from the list \eqref{eq:ex1M},
we note that a theory can only be conformally invariant if all of
its subsectors are. As it was argued for instance in
\cite{Babichenko:2006uc}, vanishing of the one loop beta function
requires that $R=0$. Hence the only candidates for conformal
quotients are to be found in the families $\mathcal{M}_\alpha(0,S)$.
The smallest subsector in these families is obtained for $\alpha_1
= S$ and $\alpha_2$ = 0, so that it takes the simple form
U(S$|$S)$/$U(S)$\times$U(S). For $S=1$, this subsector is the
theory of free symplectic fermions. In all other cases, it is a
massive theory. Hence the only candidates for conformal quotients one
can find within the list \eqref{eq:ex1M} are of the form
\begin{equation} \label{eq:EX1CFT}
\mathcal{C}^{{\text{U/U$^2$}}}_{(\alpha_1,\alpha_2)} \ : \cong\
\mathcal{M}^{{\text{U/U$^2$}}}_{(\alpha_1,\alpha_2)}(0,1)  \ \cong
\ \frac{\U{$\alpha_1$}{$\alpha_1$}}{\U{1+$\alpha_2$}{$\alpha_2$}
\times \U{$\alpha_1$-$\alpha_2$}{$\alpha_1$-$\alpha_2$-1}}
\end{equation}
with $\alpha_1 > \alpha_2 \geq 0$. Later we shall argue that the
converse is also true: symmetric superspaces that possess a non-trivial
conformal subsector with central charge $c \neq 0$ are actually
conformal. Since the theory of free symplectic fermions has
central charge $c = -2 \neq 0$, all the models in the list
\eqref{eq:EX1CFT} give rise to conformal sigma models. The list
includes the complex projective superspaces
$\mathbb{CP}^{\alpha_1-1|\alpha_1}$ for which conformal invariance
has been established before (see e.g. \cite{Sethi:1994ch}
\cite{Read:2001pz}). We shall extend this discussion to arbitrary
compact symmetric superspaces in section 5.1. Within this class,
we shall thereby recover the complete classification of conformal
models from \cite{Candu:thesis}.

But our approach is more general. It also applies to all coset
superspaces $G/G'$ without any additional assumption on the
denominator subgroup $G'$. In section 5.2 we look at examples for
which $G'$ is fixed under the action of some automorphism of order
four, i.e.\ at quotients of the form $G/G^{\mathbb{Z}_4}$. Such
generalized symmetric spaces have become popular through the
investigation of strings in Anti de Sitter backgrounds. While we
are not aiming at an exhaustive investigation of quotients within
this class, we shall exhibit a few interesting examples, including
the family
\begin{equation} \label{eq:ex2M}
\mathcal{M}^{\text{U/OSP$^2$}}_{(\alpha_1,\alpha_2)}(S)  \ \cong
\
\frac{\PSU{2$\alpha_1$}}{\OSP{2(S+$\alpha_2$)}{2$\alpha_2$} \times
\OSP{2($\alpha_1$-$\alpha_2$)}{2($\alpha_1$-$\alpha_2$-$S$)}}
\end{equation}
with some obvious restrictions on the choice of $\alpha_i$ and $S$ such that all
supergroups are well-defined. Note that, provided the $\alpha_{i}$ are large
enough, the parameter $S$ may now assume any integer value, i.e.\ it can also be
negative.
The minimal
non-trivial subsector of these theories depends significantly on the
parameter $S$. It is given by
\begin{align}
\mathcal{R}^{\text{PSU/OSP$^2$}}(S) &\  \cong\ 
\frac{\PSU{2S}}{\text{SO(2S) $\times$ SO(2S)}} \ \ \ &\mbox{
for } \ \ S\ > \ 0 \ \ , \label{eq:ex2R>} \\[2mm]
\mathcal{R}^{\text{PSU/OSP$^2$}}(0) &\  \cong\ 
\text{symplectic fermions} \ \ \ &\mbox{ for }\ \ S
\ = \ 0 \ \ , \label{eq:ex2R0} \\[2mm]
\mathcal{R}^{\text{PSU/OSP$^2$}}(S) &\  \cong\  
\frac{\PSU{-2S}}{\text{SP(-2S) $\times$ SP(-2S)}} \ \ \
&\mbox{ for } \ \ S \ < \ 0 \ \ . \label{eq:ex2R<}
\end{align}
These are not conformal for $S\neq 0$ and reduce to a free theory for $S=0$.
The smallest interacting theory for $S=0$ is obtained for $\alpha_1=1$,
$\alpha_2=0$ and is the complex projective superspace
\begin{equation}
\frac{\text {\PSU{2}}}{\OSP{2}{2}}\cong \mathbb{CP}^{1|2}\ \ .
\end{equation}
For higher values of $\alpha_i$ however, the superspaces are not of the complex
projective type. It would be interesting to understand whether the family
\eqref{eq:ex2M} with $S=0$ is conformally invariant. We have little more to say
about this issue for now.

The series \eqref{eq:ex2M} contains a few other interesting
minimal subsectors. In fact, for the $S=1$, the minimal subsector
is given in eq.\ \eqref{eq:ex2R>}. After an appropriate change in
the choice of reality conditions, we obtain the coset geometry for
$AdS_2 \times S^2$ as defined in\cite{Berkovits:1999zq}. Similarly, if we set
$S=-2$ and perform again the appropriate change of the real form, we find
the quotient that appears in the description of $AdS_5 \times S^5$.
Throughout most of this text, we shall consider sigma models
without Wess-Zumino terms, mostly in order not to clutter the
presentation too much. We shall comment on the possible inclusion
of Wess-Zumino terms and the application to other 2-dimensional
field theories in the concluding section.

We finish this introduction with a short guide for the subsequent
sections. In the next section \ref{Definitions} we shall set the
stage by defining in detail the models that we are going to
consider. Subsections \ref{sec:cr_lsa} to \ref{sec:crl2} then present the main mathematical tools at our
disposal. Since these parts are a bit technical, we included a
non-technical summary in subsection 3.1. The impatient reader may
therefore skip subsections \ref{sec:cr_lsa} to \ref{sec:crl2}, at
least upon first reading. The mathematical background from section
3 is then used in section \ref{reductionfield} to prove the main
results of this work. In section \ref{reductionexamples} we shall
illustrate how the cohomological reduction works for symmetric
superspaces. Once this is understood, we venture  into generalized
symmetric spaces. Our conclusion contains a few more comments on
possible applications to more types of models and to $AdS$
backgrounds in string theory.

\section{Sigma models on coset superspaces $G/G^{\prime}$}
\label{Definitions}

The purpose of this section is to set the stage for our subsequent
investigation. We shall provide two different formulations for
non-linear sigma models on a right-coset superspace of the form $G/G'$. Here
$G$ is some supergroup with non-degenerate metric and $G'$ is a
sub-supergroup. For the moment, no further assumption is made
concerning the structure of $G'$. In later sections, fermionic
elements of $G'$ shall play a key role. There exist some tricks to
extend the validity of our analysis in case $G'$ does not contain
any such elements. We shall discuss these briefly in case $G' =
\{e\}$ is trivial. The examples in section 5 focus on models in
which $G' = G^{\mathbb{Z}_n}$ is invariant under some automorphism
of order $n=2$ or $n=4$. But for the general framework such
special features of $G'$ are irrelevant.

\subsection{General coset superspaces $G/G'$}

We want to consider non-linear sigma models on homogeneous
superspaces $G/G^{\prime}$, where the quotient is defined as the
set of right cosets of $G'$ in $G$ through the identification
\beq g\ \sim \ gh \ \ \mbox{ for all } \ \  h\ \in \ G^{\prime}\
\subset\  G\ \ . \eeq
Let $\mathfrak{g}$ be the Lie superalgebra associated to $G$. We
assume that $\mathfrak{g}$ comes equipped with a non-degenerate
invariant bilinear form $(\phantom{x},\phantom{x})$. Examples
include $\mathfrak{g} = \gl{m}{n}$, $\sgl{m}{n}$\footnote{We
exclude $\sgl{n}{n}$ and $\pgl{n}$, since it does not have a
non-degenerate metric}, $\psl{n}$  or $\osp{m}{2n}$. Similarly,
let $\mathfrak{g}^{\prime}$ be the Lie superalgebra associated to
$G^{\prime}$. We assume that the restriction of
$(\phantom{x},\phantom{x})$ to $\mathfrak{g}'$ is non-degenerate.
In this case, the orthogonal complement $\mathfrak{m}$ of
$\mathfrak{g}'$ in $\mathfrak{g}$ is a $\mathfrak{g}'$-module and
one can write the following $\mathfrak{g}'$-module decomposition
$\mathfrak{g}=\mathfrak{g}^{\prime}\oplus \mathfrak{m}$. In
particular, this means that there are projectors $P^{\prime}$ onto
$\mathfrak{g}^{\prime}$ and $P=\mathbbm{1}-P^{\prime}$ onto $\mathfrak{m}$
which commute with the action of $\mathfrak{g}'$.

With the above requirements, the quotient $G/G'$ can be endowed
with a $G$-invariant metric $\mathsf{g}$. This metric is by no
means unique and generally depends on some number of continuous
parameters which we shall also call radii. The square root of the
superdeterminant of $\mathsf{g}$ provides in the standard way a
$G$-invariant measure $\mu$ on $G/G'$. This measure is unique up
to a multiplicative constant which depends on the radii of the
metric $\mathsf{g}$. With these two structures one can already
write down a purely  kinetic Lagrangian for the sigma model on
$G/G'$ and quantize it in the path integral formalism. Inclusion
of $\theta$-terms,  WZW terms or $B$-fields requires a better
understanding of the geometry of the $G/G'$ superspace. In fact,
the $\theta$ and WZW terms are associated to $G$-invariant closed
but not exact 2- and 3-forms, respectively. $B$-fields, on the
other hand, are written in terms of $G$-invariant exact 2-forms.
Every such linearly independent form comes with its own coupling
constant. We shall only consider Lagrangians with a kinetic term
and a $B$-field. Let $\mathsf{b}$ be some general linear
combination of $G$-invariant exact 2-forms. Then the most general
Lagrangian we consider can be written in the form
\begin{equation}\label{eq:lag_origin}
 \mathcal{L} = \eta^{\mu\nu}\mathsf{g}
 (\partial_\mu,\partial_\nu)+\epsilon^{\mu\nu}\mathsf{b}(\partial_\mu,
\partial_\nu) \ ,
\end{equation}
where $\eta^{\mu\nu}$ is the constant world sheet metric,
$\epsilon^{\mu\nu}$ the antisymmetric tensor with
$\epsilon^{01}=1$. The Lagrangian is obviously evaluated on maps
from the worldsheet $\Sigma$ to the superspace $G/G'$ and to every
one of such maps one can associate a  vector field $\partial_\mu$
on $G/G'$, which appears in eq.~\eqref{eq:lag_origin} in a
coordinate free notation.

There is a different way to formulate the sigma model on $G/G'$,
which makes its coset nature manifest and allows to explicitly
construct the metric $\mathsf{g}$ and the $B$-field $\mathsf{b}$
in eq.~\eqref{eq:lag_origin}. For that purpose, instead of maps
from the worldsheet to the target space $G/G'$, we consider more
general maps $g:\Sigma\rightarrow G$ from the world sheet to the
Lie supergroup $G$. A basis set of  1-forms on $G$ which are
invariant under the global left $G$-action is provided by the so
called Maurer-Cartan forms
\beq \label{eq:currdef} J_{\mu}(x)\ =\ g^{-1}(x)\partial_{\mu}
g(x)\ . \eeq
Higher $G$-invariant tensors may be built out of the Maurer-Cartan
forms by taking tensor products. There is a subspace of such
tensors which are also invariant with respect to the \emph{local}
right $G'$-action. These may be specified by their values on the
coset superspace $G/G'$. We use this idea in order to build
explicitly the $G$-invariant tensors $\mathsf{g}$ and $\mathsf{b}$
that enter the Lagrangian~\eqref{eq:lag_origin}.

Under right $G'$-gauge transformations $g^{\prime}:\Sigma\mapsto
G^{\prime}$ the Maurer-Cartan forms $J_\mu$ transform as
\beq \label{rightaction} g(x)\ \mapsto\  g(x)g^{\prime}(x)\qquad
J_{\mu}(x)\ \mapsto\
(g^{\prime}(x))^{-1}J_{\mu}(x)g^{\prime}(x)+(g^{\prime}(x))^{-1}\partial_{\mu
} g^{\prime}(x)\ . \eeq
Since the projection $P$ on $\mathfrak{m}$ commutes with the
action of $\mathfrak{g}^{\prime}$, the projected forms $P(J_{\mu})$
transforms by conjugation with $g^\prime$. To build right
$G'$-gauge invariant 2-forms we introduce the
$\mathfrak{g}'$-intertwiners
\begin{equation}\label{eq:GBdef}
\tG\ \in\
\text{End}_{\mathfrak{g}^{\prime}}\left(\mathfrak{m}\circ\mathfrak{m},
\mathbb{C}\right) \ \ \ \mbox{ and } \ \ \  \tB\ \in\
\text{End}_{\mathfrak{g}^{\prime}}\left(\mathfrak{m}\wedge\mathfrak{m},
\mathbb{C}\right)
\end{equation}
from the symmetric,
respectively antisymmetric tensor product of $\mathfrak{m}$ with
itself to the trivial representation.
In terms of these intertwiners the Lagrangian~\eqref{eq:lag_origin} takes the
explicit form
\beq \label{generalLagrangian} \calL\ =\
\eta^{\mu\nu}\tG\left(P(J_{\mu}),P(J_{\nu})\right)+\epsilon^{
\mu\nu } \tB\left(P(J_{\mu}),P(J_{\nu})\right)\ . \eeq
The  choice of $\tG$ and $\tB$, subject to some reality
constraints, parametrizes the moduli space of the sigma model on
$G/G'$ with a kinetic term and a $B$-field only. Global left
$G$-invariance of the Lagrangian~\eqref{generalLagrangian} is
automatic since Maurer-Cartan forms $J_{\mu}(x)$ are left
$G$-invariant by construction. Right $G^{\prime}$-gauge
invariance, on the other hand, follows easily from the
transformation properties of $P(J_\mu)$ and the
def.~\eqref{eq:GBdef} of $\tG$ and $\tB$ as invariant bilinear
forms on the $\mathfrak{g}'$-module
$\mathfrak{m}\otimes\mathfrak{m}$.

\subsection{$G/G^{\mathbb{Z}_N}$ coset superspaces}

In the previous subsection we have described the most general
action with a kinetic term and a $B$-field for the $G$-invariant
sigma model with target space $G/G'$. The formulation includes
sigma models on symmetric spaces and certain generalizations that
appear in the context of AdS compactifications. In fact, for many
cases of interest, the Lie sub-superalgebra $\mathfrak{g}^{\prime}$ in
$\mathfrak{g}$ consists of elements that are invariant under some
finite order automorphism $\Omega:\mathfrak{g} \mapsto \mathfrak{g}$. An
automorphism of order $N$ defines a decomposition
\begin{equation} \label{eq:decOmeig}
\mathfrak{g}\ =\ \mathfrak{g}^{\prime}\oplus
\bigoplus_{i=1}^{N-1}\mathfrak{m}_i\qquad , \quad
\Omega|_{\mathfrak{g}^{\prime}}\ =\ \mathbbm{1}\quad , \quad
\Omega(\mathfrak{m}_k)\ =\ e^{\frac{2\pi i k}{N}}\mathfrak{m}_k \end{equation}
of the superalgebra $\mathfrak{g}$ into eigenspaces of $\Omega$.
Extending our previous notation, we denote by $P_i$ the projection
maps onto $\mathfrak{m}_i$. Thanks to the properties of the
$\Omega$, we find
\beq \label{properties1} \left[\mathfrak{m}_i,\mathfrak{m}_j\right]\
\subset\  \mathfrak{m}_{i+j\text{ mod } N}\qquad
\left(\mathfrak{m}_i,\mathfrak{m}_j\right)\ =\ 0\ \ \text{ if } \ \ i+j\
\neq \ 0 \text{ mod } N\ , \eeq
where we have set $\mathfrak{m}_0\equiv \mathfrak{g}^{\prime}$.
Consequently, the subalgebra $\mathfrak{g}^{\prime}$ acts on the
$\Omega$-eigenspaces $\mathfrak{m}_i$. Note that the spaces $\mathfrak{m}_i$
need
not be indecomposable under $\mathfrak{g}^{\prime}$ in which case the
decomposition into $\mathfrak{g}^\prime$-modules is finer than the
decomposition \eqref{eq:decOmeig} into eigenspaces of $\Omega$.

Whenever a coset superspaces $G/G'$ is defined by an automorphism
$\Omega$ of order $N$ we shall use the alternative notation
$G/G^{\mathbb{Z}_N}$. The cases when the grading induced by
$\Omega$ is compatible with the $\mathbb{Z}_2$ superalgebra
grading, that is $\mathfrak{m}_{2i}\in \mathfrak{g}_{\bar{0}}$ and
$\mathfrak{m}_{2i-1}\in \mathfrak{g}_{\bar{1}}$, were considered
by Kagan and Young in \cite{Kagan:2005wt}. They restricted to a
family of Lagrangians for which $\tG$ and $\tB$ take the following
special form
\beq \label{eq:KYGB} \tG(X,Y)\ =\
\sum_{i=1}^{N-1}p_i\left(P_i(X),P_{N-i}(Y)\right)\quad , \quad
\tB(X,Y)\ =\ \sum_{i=1}^{N-1}q_i\left(P_i(X),P_{N-i}(Y)\right)\ ,
\eeq
where  the $p_i$ and $q_i$ are constants obeying the additional
constraints
\beq p_i=p_{N-i}\qquad q_i=-q_{N-i}\ . \eeq
The forms of $\tG$ and $\tB$ in eq.~\eqref{eq:KYGB} do not give
rise to the most general Lagrangian for coset superspaces $G/G'$.
As an example consider the famous $\mathbb{Z}_4$ quotient
$\text{PSU}{(2,2|4)}/\BSO{1,4}\times\BSO{5}$.
Its metric has two radii because its bosonic base is
 $AdS_5\times S^5$.
On the other hand, the special form of $\tG$ in
eq.~\eqref{eq:KYGB} allows for only two parameters $p_1=p_3$ and
$p_2$, among which $p_1$ is redundant because of the purely
fermionic nature of $\mathfrak{m}_1$ and $\mathfrak{m}_3$. In this
example, the form that $\tG$ takes in eq.~\eqref{eq:KYGB}
restricts the radii of $AdS_5$ and $S^5$ to be equal.

The properties of the theory defined by eqs.\ \eqref{eq:KYGB}
certainly depend on the precise choice of the parameters $p_i$
and $q_i$. In particular, it was shown in \cite{Young:2005jv} and
\cite{Kagan:2005wt} that one loop conformal invariance requires
\beq p_i\ =\ 1\qquad q_i\ =\ 1-\frac{2i}{N}\qquad \text{ for } i\
\neq\  0\ , \eeq
for all even $N$. We believe, however, that in most cases these
conditions are not sufficient to guarantee the vanishing of the
full beta function.

Our second comment concerns the treatment of coset superspaces $G/G'$
in which the denominator group $G'$ has a non-trivial centralizer
$Z \subset G$. For such coset superspaces, there exists a
residual symmetry by right multiplications with elements of $Z$.
In an equivalent formulation one can make all symmetries of $G/G'$ to act from
the left.
For that we rewrite $G/G' = G \times Z/ G' \times Z$ where the
factor $Z$ in the denominator is embedded diagonally into the
numerator.
To make the associated reformulation of the sigma model
a bit more explicit, we focus on the principal chiral model for
the supergroup $U$. Without any further thought one might be
tempted to describe this model through $G = U$ and $G' = \{ e\}$.
But as our introductory comments suggest, we prefer to rewrite the
group manifold $U$ as a coset superspace $U = U \times U/U$ and hence to set
\beq G\ =\ \left\{\left(x,y\right): x,y \in  U\right\}\qquad ,
\qquad G^{\prime}\ =\ \left\{\left(x,x\right): x\in U\right\}\ .
\eeq
The left and right action of $G$ on itself is given by
componentwise multiplication. The right coset superspace
$G/G^{\prime}\cong U$ is considered as the space of equivalence
classes under the equivalence relation $\left(x,y\right)\sim
\left(xz,yz\right)$, for all $z\in U$. In particular,
$\left(xy^{-1},1\right)$ is the canonical representative of the
equivalence class of $\left(x,y\right)$. Hence, the currents
$J_{\mu}$ and the projection map $P:\mathfrak{g}\rightarrow \mathfrak{m}$
are given by
\beq J_{\mu}\ =\
\left(x^{-1}\partial_{\mu}x,y^{-1}\partial_{\mu}y\right)\qquad ,
\qquad  P:\left(v,w\right)\ \mapsto\
\left(\frac{v-w}{2},-\frac{v-w}{2}\right)\ . \eeq
If $(\phantom{x},\phantom{x})$ is the invariant form on the Lie superalgebra
of $U$ and we take $\tG$ to be given by
\beq \tG\left(\left(v_1,w_1\right)\circ
\left(v_2,w_2\right)\right)\ = \ \left(v_1,v_2\right)+
\left(w_1,w_2\right)\  \eeq
we obtain the usual principal chiral model for $U$. In fact, one
may easily show that
$$\tG\left(P\left(J_{\mu}\right),P\left(J_{\nu}\right)\right)\eta^{\mu\nu}
\ =\ \frac{1}{2}\left(u^{-1}\partial_{\mu} u, u^{-1}\partial_{\nu}
u\right)\eta^{\mu\nu}\ \ , $$where $u=xy^{-1}\in U$. Thereby we
have established the standard geometric results that allows us to
treat the principal chiral model on $U$ as a $G/G'$ coset
superspace model. The advantage of the seemingly more complicated
coset description will become apparent below.

\subsection{Observables and correlators}
\label{sec:observables}

We give a brief description of observables
and their correlation functions. Let us denote by ${\mathcal G}$
and ${\mathcal G}^\prime$ the space of all continuous maps from
the world-sheet $\Sigma$ to the supergroups $G$ and $G'$,
respectively. Obviously, $\mathcal{G}'$ acts on $\mathcal{G}$ by
point-wise (on $\Sigma$) right multiplication. Local observables
of the $G/G'$ quotient model are defined as some well behaved
class of maps $\mathcal{O}:\mathcal{G}\times \Sigma\mapsto
\mathbb{C}$ invariant by this right $\mathcal{G}'$ action
\beq \mathcal{F}_{G/G^{\prime}}\ =\ \left\{\ {\cal{O}}:
\mathcal{G}\times \Sigma \mapsto \mathbb{C} \ | \  \mathcal{O}(g,x) \ = \
{\cal{O}} (g\cdot g^{\prime},x) \ \mbox{for all} \  g^{\prime} \in
\mathcal {G}^{\prime} \right\} \ , \eeq
where we have denoted $\mathcal{O}(g,x):= \mathcal{O}(g(x))$.

One class of observables is obtained by restricting smooth right
$G'$-invariant functions $f:G\mapsto \mathbb{C}$ to the image of
an arbitrary map $g:\Sigma\mapsto G$. Existence of the 2-point
function for this observable $f(g(x))$ requires that $f\in
L_2(G/G')$. These are the tachyonic fields.

Similarly, all other observables can be obtained from smooth right
$G'$-invariant tensor forms $t$ of rank $k$ on $G$ by restricting
them to the image of some arbitrary map $g:\Sigma\mapsto G$ and
evaluating them on the set of vector fields
$\partial_{\mu_1},\dots,\partial_{\mu_k}$. Existence of
correlation functions for the observables
$t_{g(x)}(\partial_{\mu_1},\dots,\partial_{\mu_k})$ imposes some
further constraints. As an example, let us consider the
Maurer-Cartan forms $J_{\mu}$ we have introduced in
eq.~\eqref{eq:currdef}. Their components do not give rise to
observables of the quotient model because there are not right
$G'$-gauge invariant. Nevertheless, recalling their
behavior~\eqref{rightaction} under right $G'$-gauge
transformations, one can build the following observables
$$ j_\mu \ = \    g P(J_\mu) g^{-1}\ \in \ \mathcal{F}_{G/G'}\ . $$
These are the Noether currents for the global symmetry $G$ of the $G/G'$ sigma model.

In the following we shall denote by $\mathcal{O} (x)$ the
restriction of the local observable $\calO$ to the point $x$ of the world-sheet. Given any set $\calO_i\in
{\cal F}_{G/G^{\prime}}$ of such local observables we define their
unnormalized correlation functions through
\beq \label{eq:def_corr_func}
\vac{\prod_{i=1}^N\calO_i(x_i)}_{G/G^{\prime}}=\int_{
\mathcal{G}} \, [\text{d}\mu_G] \ e^{-S}\
\prod_{i=1}^N\calO_i(x_i)\ . \eeq
Here, $S=\int_\Sigma d^2 x\, \mathcal{L}$ is the action
~\eqref{generalLagrangian} of our model. Our definition of
correlation functions involves an integration over elements of
${\cal G}$ with some left ${\cal G}$-invariant measure
\begin{equation}
 [d \mu_G(g)] = \prod_{x\in\Sigma}d\mu_G(g(x))\ ,
\end{equation}
where $d\mu_G$ is the unique (up to normalization) Haar measure on
$G$. In eq.~\eqref{eq:def_corr_func}, the integration over $G$ at
every point of the worldsheet yields a factor which is the volume
of $G'$. Strictly speaking, this makes sense only if $G'$ is
compact. We assume that the contribution of such factors can be
properly regularized and renormalized by replacing the worldsheet
$\Sigma$ with a lattice,  and shall not dwell on such details.

The reader might be curious about why we insist on integrating
over maps $\mathcal{G}$ from the worldsheet to the group $G$
rather then maps from the worldsheet to the quotient $G/G'$. In
other words, why we do not fix the right $G'$-gauge invariance? As
we shall see, keeping this symmetry explicit in the quantum theory
simplifies the cohomology calculations on tensor fields.

\section{Cohomological reduction in representation theory}

The following section contains most of the mathematical results we
shall need below. Since several of our statements seem to be new,
we decided to present and prove them in a rather mathematical
style. For pedagogical reasons, however,  we shall begin with a
short overview of the most relevant notations and results. This
should enable impatient readers to skip over subsections 3.2 -- 3.6,
at least upon first reading.

\def\g{\mathfrak{g}}
\def\h{\mathfrak{h}}
\def\e{\mathfrak{e}}
\def\f{\mathfrak{f}}
\def\m{\mathfrak{m}}
\newcommand\CPone{\ensuremath{\text{$\mathbb{CP}$}^{1|2}
}}

\subsection{Overview over results}\label{sec:into_math}

As in the previous subsection we assume $\g$ to be a Lie
superalgebra with a non-degenerate symmetric bilinear form
$(\phantom{x},\phantom{x})$. Let us pick some fermionic element $Q \in \g$
that squares to zero, i.e.\ $[Q,Q] = 2 Q^2 = 0$. Such elements
exist for most Lie superalgebras of interest, with the exception of the series
osp(1$|$2N). The element $Q$ defines a decomposition of $\g$ into
three Lie sub-superalgebras $\h,\e,\f$,
\begin{eqnarray*}
\g & = & \h \oplus \e \oplus \f \ \ \quad \quad \mbox{ such that} \\[2mm]
\e & = & \Img_Q \g \ \ \mbox{ and } \ \ \h \oplus \e  =  \Ker_Q \g
\ \ .
\end{eqnarray*}
The bilinear form $(\phantom{x},\phantom{x})$ restricts to a non-degenerate
form on $\h \subset \g$. The Lie sub-superalgebras $\e$ and $\f$,
on the other hand, are isotropic, i.e.\ $ (\e , \e) = 0 = (\f ,
\f)$. We also note that $\e$ and $\f$ both carry an action of
the Lie superalgebra $\h$.

In subsection 3.2 we shall compute the Lie superalgebra $\h$ for
various choices of $\g$ and any $Q \in \g$. The results may be
summarized as follows
\begin{eqnarray}
 \h(\gl{M}{N}) & \simeq & \gl{M-r_Q}{N-r_Q}\ ,  \\[2mm]
  \h(\sgl{M}{N}) & \simeq & \sgl{M-r_Q}{N-r_Q}\ ,
  \\[2mm]
\h(\osp{R}{2N}) & \simeq &  \osp{R-2 r_Q}{2N- 2 r_Q}\ .
\end{eqnarray}
The answer depends on $Q$ only through an integer $\rank{(Q)} =
r_Q \geq 1$ that will be defined in section 3.2. In all three
cases we listed above, there exist elements $Q$ with minimal rank
$r_Q =1$.

The element $Q$ acts in any representation $V$ of $\g$ and defines
the following cohomology classes
$$ \Coh_Q(V) \ =\ \Ker_Q V / \Img_Q V \ \ . $$
The linear space $\Coh_Q(V)$ comes equipped with an action of the
Lie sub-superalgebra $\h \subset \g$. It is not difficult to see
(cp. section 3.3) that $V \rightarrow \Coh_Q(V)$ is functorial,
i.e.\ it is consistent with forming tensor products, direct sums
and conjugation in the category of $\h$-representations.

Though $\Coh_Q(V)$ vanishes for a large class of representations
(see below), it can certainly contain non-trivial elements. Note,
for example, that the cohomology of the adjoint $\g$-module $V =
\g$ is given by $\Coh_Q(\g) = \h$. One may actually show that $V$
and $\Coh_Q(V)$ possess the same super-dimension. Hence, all
representations $V$ with non-vanishing super-dimension $\sdim V$ =
$\dim V_{\bar 0}$ - $\dim V_{\bar 1}$ must give rise to $\Coh_Q(V)
\neq 0$. The condition $\sdim V \neq 0$ is often satisfied for
short multiplets (atypical representations). For long (typical)
multiplets $V$, on the other hand, the cohomology $\Coh_Q(V)$ is
always trivial. More generally, we will see that $\Coh_Q(V) = 0$
for all (finite dimensional) projective modules.

Let us now consider a Lie superalgebra $\g$ along with a
subalgebra $\g' \subset \g$. The corresponding Lie supergroups
will be denoted by $G$ and $G'$, respectively. As before, we want
to pick some fermionic element $Q \in \g$ with $Q^2 =0$. Let us
now assume that $Q$ is contained in the subalgebra $\g' \subset
\g$ so that its cohomology defines two Lie sub-superalgebras $\h
\subset \g$ and $\h' \subset \g'$ with $\h' \subset \h$. We denote
the associated Lie supergroups by $H$ and $H'$, respectively. Note
that the space of functions on the coset superspace $G/G'$ carries an
action of $\g$. In particular, the element $Q$ acts and gives rise
to some cohomology. The central claim of this section is that
the cohomology of some geometric object (smooth function, tensor form, square integrable function) defined on the coset superspace $G/G'$ is equivalent to a similar object defined on $H/H'$.
This gives rise to isomorphisms of the type
\begin{equation} \label{eq:mainresult}
\Coh_Q(L_2(G/G')) \ \cong \ L_2(H/H')\ \ ,
\end{equation}
which means that the cohomology of $Q$ in the space of square
integrable functions on $G/G'$ may be interpreted as a space of
square integrable functions on the coset superspace $H/H'$. We
note that $L_2(H/H')$ carries an action of the Lie superalgebra
$\h = \Coh_Q(\g) \subset \g$. The isomorphism
\eqref{eq:mainresult} is an isomorphism of $\h$ modules.

The derivation of eq.\ \eqref{eq:mainresult} is a bit involved. We
shall provide a fully explicit proof in section 3.6. Here, we
shall content ourselves with some more qualitative arguments.  By
construction, $\Coh_Q(L_2(G/G'))$ is a commutative algebra and
hence it can be considered as an algebra of functions on some
space $X$. The latter is acted upon by the supergroup $H$ with Lie
superalgebra $\Coh_Q(\g) = \h$. Since the action of $G$ on $G/G'$
is transitive, it suffices to understand the reduction from $G/G'$
to $X$ locally, near the image $eG'\in G/G'$ of the group unit
$e\in G$. The tangent space at this point of the coset supermanifold is given by
$\g/\g' \equiv \m$. Its cohomology is given by
\begin{equation}
\Coh_Q(\m) \ = \ \Coh_Q(\g/\g') \ = \ \Coh_Q(\g)/\Coh_Q(\g') \ = \
\h/\h'\ ,
\end{equation}
i.e.\ the tangent vectors to the reduced space $X$ lie in
$\h/\h'$. Thereby we conclude that $X = H/H'$. Now, let $\langle
\phantom{x},\phantom{x}\rangle_{G/G'}$ be the $G$-invariant scalar
product of $L_2(G/G')$. It is very easy to see that $\langle
\phantom{x},\phantom{x}\rangle_{G/G'}$ descends to cohomology.
Hence, the space $\Coh_Q(L_2(G/G'))$ of functions inherits an
$L_2$ structure from $L_2(G/G')$. We shall denote it by
$\langle\phantom{x},\phantom{x}\rangle_{H/H'}$. Its $H$-invariance
follows immediately from the $G$-invariance of $\langle
\phantom{x},\phantom{x}\rangle_{G/G'}$ and the inclusion
$\mathfrak{h}\subset\Ker_Q\mathfrak{g}$. General results on
measure theory \cite{measure} then imply that the scalar product
$\langle\phantom{x},\phantom{x}\rangle_{H/H'}$ arises from a
measure on $H/H'$, which is unique (up to a constant factor) by
$H$-invariance. Hence, we have established
eq.~\eqref{eq:mainresult}.

As an example of the above, let us discuss the Lie superalgebra
$\g$ = gl(2$|$2). For $Q$ we pick the supermatrix that contains a
single entry in the upper right corner. It is then easy to check
that
$$
\Ker_Q\g \ = \ \h \oplus \e \ \ni \ \left( \begin{array}{cccc}
a_{11} &
a_{12} & b_{11} & b_{12} \\ 0 & a_{22} & b_{21} & b_{22} \\
0 & c_{12} & d_{11} & d_{12} \\
0 & 0 & 0 & a_{11} \end{array} \right) \ , \ \  \Img_Q \g \ = \ \e
\ \ni \ \left( \begin{array}{cccc} a_{11} &
a_{12} & b_{11} & b_{12} \\ 0 & 0 & 0 & b_{22} \\
0 & 0 & 0 & d_{12} \\
0 & 0 & 0 & a_{11} \end{array} \right)\ .
$$
Consequently,  $\Coh_Q(\g) = \h = $ gl(1$|$1) consists of all
supermatrices in which $a_{22}, b_{21}, c_{12}$ and $d_{11}$ are
the only non-vanishing entries. Let us also specify the Lie
sub-superalgebra $\g'$ to consist of all elements in $\g$ with
vanishing entries $b_{11} = b_{21} = d_{12} = d_{21} = c_{11} =
c_{12} = 0$. Hence, $\g' \cong$ gl(2$|$1) $\times$ gl(1). The
cohomology $\Coh_Q(\g') = \h' =$ gl(1) $\times$ gl(1) of $\g'$ can
be read off easily.

In our example, the quotient  $G/G'$ is the complex projective
superspace $\CPone  \cong S^2 \times \mathbb{R}^{0|4}$. Functions
thereon may be decomposed into finite dimensional representations
of psl(2$|$2) as follows
$$ L_2(\CPone) \ \cong \ \bigoplus_{j=0}^{\infty}\ \ [j,0]\ \ .  $$
The representations $[j,0]$ of psl(2$|$2) that appear in this
decomposition possess dimension $d_j = 16(2j+1)$. They are
generated from the spherical harmonics on the bosonic 2-sphere by
application of four fermionic generators. For $j \neq 0$, the
psl(2$|$2) modules $[j,0]$ turn out to be projective (typical long
multiplets) and hence $\Coh_Q([j,0]) = 0$ for all $j \neq 0$. The
only non-vanishing cohomology comes from the 16-dimensional Kac
module $[0,0]$. The latter is built from three atypical
irreducibles, namely two copies of the trivial representation and
one copy of the 14-dimensional adjoint representation of
psl(2$|$2). Each of these pieces contributes to cohomology. While
the two trivial representations give rise to two even states, the
adjoint representation has an excess of two odd states which
descend to cohomology. In total, we obtain a 4-dimensional
cohomology
$$ \Coh_Q(L_2(G/G')) \ = \ \Coh_Q(L_2(\CPone)) \ = \ \Coh_Q([0,0])
\ = \ \mathbb{R}^{2|2} \ \ . $$ To me more precise, we note that
the linear space $\mathbb{R}^{2|2}$ carries the 4-dimensional
projective cover of gl(1$|$1). According to our general statement,
the cohomology should agree with the space of functions on the
quotient $H/H' =$ GL(1$|$1)/GL(1) $\times$ GL(1). The quotient
possesses two fermionic coordinates and hence gives rise to a
4-dimensional algebra of functions over it,
$$ L_2(H/H') \ = \ \mathbb{R}^{2|2} \ \ . $$
It indeed agrees with the cohomology in the space of functions
over $\CPone$, as it was claimed in eq.\ \eqref{eq:mainresult}.

\subsection{Reduction of Lie superalgebras}
\label{sec:cr_lsa}

%
%
As in the previous subsection, let $\mathfrak{g}$ denote a Lie
superalgebra and $Q$ be any fermionic element of $\mathfrak{g}$
with vanishing bracket, that is $[Q,Q]=2Q^2=0$.

\begin{lemma} The element $Q \in \g$ gives rise to a linear map
$Q: \g \rightarrow \g$ that is defined by $Q(X) = [Q,X]$ for all
$X \in \g$. Then it is possible to show that
\begin{enumerate}
  \item[1)] the subspaces $\Ker_Q\mathfrak{g}$ and $\Img_Q \mathfrak{g}$
  are subalgebras of $\mathfrak{g}$,
  \item[2)] the subalgebra $\Img_Q \mathfrak{g}$ is an ideal of $\Ker_Q \mathfrak{g}$,
  \item[3)] the quotient space $\Coh_Q(\mathfrak{g})$ is a Lie superalgebra.
\end{enumerate}
\end{lemma}
All assertions of this lemma are easily established using no more
that the (graded) Jacobi identity. The Lie bracket on the quotient
space  $\Coh_Q(\mathfrak{g})$ is induced from the Lie bracket of
$\mathfrak{g}$ through
\begin{equation}
  [x+\Img_Q \mathfrak{g},y+\Img_Q \mathfrak{g}] \ = \ [x,y]+\Img_Q
  \mathfrak{g},\qquad x,y\in \Ker_Q\mathfrak{g} \ .
\end{equation}
We shall often refer to the space $\Coh_Q(\mathfrak{g})$ as the
{\em cohomological reduction} of the Lie superalgebra
$\mathfrak{g}$ with respect to $Q$. In our discussion of concrete
examples we shall essentially restrict to the superalgebras
$\mathfrak{g}$ of the type $\osp{M}{2N}$, $\gl{M}{N}$ or
$\sgl{M}{N}, N \neq M$. All these Lie superalgebras possess an
invariant, consistent, supersymmetric, non-degenerate bilinear
form $(\phantom{x},\phantom{x}):
\mathfrak{g}\times\mathfrak{g}\rightarrow \mathbb{C}$.

The adjoint action of $Q$ can be brought in its Jordan normal form by choosing a
basis $\left\{h_a\right\}\cup\left\{e_i,
f_i\right\}$ of $\mathfrak{g}$ such that
\beq \label{eq:start_mark} \left[Q,h_a\right]\ =\ 0 \quad \text{
and } \quad \left[Q,f_i\right]\ =\ e_i\ . \eeq
Using the invariance of the bilinear form we show that
\begin{equation}\label{eq:ort1}
 (h_a,e_i)\ =\ 0\ ,\qquad
 (e_i,e_j)\ =\ 0\ .
\end{equation}
If follows from the non-degeneracy of the bilinear form that the
matrix $D_{ij}=\left(e_i,f_j\right)$ must be invertible. Defining
\begin{align}
h_a^{\prime}\ =\
h_a-\left(h_a,f_j\right)\left(D^{-1}\right)^{ji}e_i\ ,
\\[2mm]
f_i^{\prime}\ =\ f_i -
\frac{1}{2}(f_i,f_j)\left(D^{-1}\right)^{jk}e_k
\end{align}
we see that
\begin{equation}\label{eq:ort2}
  \left(h_a^{\prime}, f'_i\right)\ =\ 0\ ,\qquad
  \left(f_i^{\prime}, f'_j\right)\ =\ 0\ \ .
\end{equation}
To prove the second assertion in eq.\ \eqref{eq:ort2} we have used
the following property of the matrix $D$
\begin{equation}
  D_{ij} \, =\,  \left([Q,f_i],f_j\right)\, =\,
  -(-1)^{|f_i|}\left(f_i,[Q,f_j]\right)\, =\,
  -(-1)^{|f_i|}\left(f_i,e_j\right) \, =\,  - D_{ji},
\end{equation}
where the last equality in the chain uses the consistency of the
bilinear form.

Let us denote by $\mathfrak{h}$, $\mathfrak{e}$ and $\mathfrak{f}$
the span of $h_a^\prime$, $e_i$ and $f_i^\prime$, respectively.
Notice that $Q$ still remains in a Jordan normal form with respect
to the new basis $h_a^\prime, e_i, f_j^\prime$. From the
eqs.~(\ref{eq:ort1}, \ref{eq:ort2}) we deduce the following
orthogonality conditions
\begin{equation}\label{eq:ort_all}
  (\mathfrak{h},\mathfrak{e}) \, =\,  (\mathfrak{h},\mathfrak{f})\,
  =\,
(\mathfrak{e},\mathfrak{e}) \, =\, (\mathfrak{f},\mathfrak{f})\,
=\, 0\ .
\end{equation}
Using once more the invariance of the bilinear form it is not hard
to derive the following features of the Lie bracket on
$\mathfrak{g}$,
    \begin{align}\label{eq:1_fam_comm}
     [\mathfrak{h},\mathfrak{h}]\, \subset\,  \mathfrak{h},&&
     [\mathfrak{h},\mathfrak{e}]\, \subset\, \mathfrak{e},&&
     [\mathfrak{h},\mathfrak{f}]\, \subset\,  \mathfrak{f},\\[2mm] \notag
     [\mathfrak{e},\mathfrak{e}]\, \subset\, \mathfrak{e},&&
     [\mathfrak{f},\mathfrak{f}]\, \subset\, \mathfrak{f},&&
     [\mathfrak{e},\mathfrak{f}]\, \subset\, \mathfrak{g} \ .
    \end{align}
Notice, in particular, that both $\mathfrak{e}$ and $\mathfrak{f}$
provide some representation for the Lie superalgebra
$\mathfrak{h}$. Furthermore, we observe that $\g$ and $\h$ possess
the same cohomology, $\Coh_Q (\mathfrak{g}) =
\Coh_Q(\mathfrak{h})$. Next, let us define the projection map
$p_{\mathfrak{h}}:\mathfrak{g}\rightarrow \mathfrak{h}$ through
\begin{equation}
  p_{\mathfrak{h}}(x)\  = \ (x,h^\prime_a)(B^{-1})^{ab}h'_b\ ,
\end{equation}
where $x\in\Ker_Q$ and $B_{ab}=(h^\prime_a,h^\prime_b)$.
The kernel of $p_{\mathfrak{h}}$ being exactly $\mathfrak{e}$,
the map $p{_\mathfrak{h}}$ is effectively defined on $\Coh_Q(\mathfrak{g})$.
Taking into account eqs.~\eqref{eq:1_fam_comm}, we see that
$p_{\mathfrak{h}}$ provides the following algebra isomorphism
\begin{equation}\label{eq:coh_isom}
  \mathfrak{h} \ \simeq\  \Coh_Q (\mathfrak{g})\ .
\end{equation}
%
In the same spirit, one can define $\mathfrak{h}$-module
projection homomorphisms $p_{\mathfrak{e}}$
and $p_{\mathfrak{f}}$ from $\mathfrak{g}$ to
$\mathfrak{e}$ and $\mathfrak{f}$, respectively,
\begin{align}\label{eq:pr_drd_alg}
  p_{\mathfrak{e}}(x) &=\, (x,f'_i) (D^{-1})^{ij}e_j\\[1mm] \notag
  p_{\mathfrak{f}}(x) &=\,  x - p_{\mathfrak{h}}(x)- p_{\mathfrak{e}}(x) \ .
\end{align}
These provide the following direct sum decomposition of
$\mathfrak{g}$,
\begin{equation}\label{eq:dir_sum_dec_coh_alg}
   \mathfrak{g}\simeq \mathfrak{h}\oplus\mathfrak{e}\oplus\mathfrak{f} \ .
\end{equation}
The isomorphism respects the action of $\h$, i.e.\ it is an
isomorphism of $\h$ modules.
\medskip

The superalgebras we consider are characterized by a Cartan
subalgebra which we denote, in a somewhat non-standard way, by
$\mathfrak{g}_0$ and a root system $\Delta$. If
$R:\mathfrak{g}\rightarrow \bgl{V}$ is the fundamental
representation, then the Cartan subalgebra $\mathfrak{g}_0$ can be
represented through diagonal matrices of $\bgl{V}$, while $\Delta$
is a subset of the root system of $\bgl{V}$.

Let us now perform the cohomological reduction for the Lie
superalgebra $\mathfrak{g}$ when $Q$ is a root generator of root
$q$ such that $2q\notin \Delta$. Consider
the root decomposition of $\mathfrak{g}$
\begin{equation}\label{eq:root_decomp}
  \mathfrak{g} \ = \ \mathfrak{g}_0 \oplus
  \bigoplus_{\alpha\in\Delta}\ \mathfrak{g}_\alpha \ .
\end{equation}
The superalgebras $\mathfrak{e}$ and $\mathfrak{f}$
can be easily evaluated
\begin{align}
  \mathfrak{e} &= \ \mathbb{C}H_q \oplus \bigoplus_{\alpha-q\in
  \Delta} \mathfrak{g}_{\alpha}\ ,\\[2mm]
  \mathfrak{f} &=\  \mathfrak{g}_0/\Ker q \oplus \bigoplus_{\alpha+q\in
  \Delta} \mathfrak{g}_{\alpha} \ ,
\end{align}
where for any weight $\lambda$ one denotes by $H_\lambda$ the
Cartan generator constructed through
\begin{equation}\label{eq:useful_trick}
  \lambda(H)\ =\ (H_\lambda,H)\ .
\end{equation}
Therefore, we can write the cohomology of $\g$ in the form
\begin{equation}\label{eq:one_rank_case}
 \Coh_q(\mathfrak{g})\, :=\, \Coh_Q(\mathfrak{g})\, \simeq\, \mathfrak{h}\, =\,  \Ker
q/\mathbb{C}H_q
 \oplus \bigoplus_{ \alpha\pm q\notin \Delta}
 \mathfrak{g}_\alpha \ .
\end{equation}
Let us apply this general result to compute the cohomological
reduction of the superalgebras gl, sl and osp. For the readers
convenience we have listed the relevant root systems in
tab.~\ref{tab:root_systems}.

\begin{table}[t]\label{tab:root_systems}
\beq
\begin{array}{lcc}
\mathfrak{g} & \Delta_{\bar{0}} & \Delta_{\bar{1}}\vspace{2pt}\\\hline\\
\gl{M}{N},\ \sgl{M}{N} & \begin{array}{c}\epsilon_i-\epsilon_j\\\delta_k-\delta_l\end{array} &
\pm\epsilon_i\mp \delta_k\vspace{6pt}\\
\osp{2M}{2N} &  \begin{array}{c}\pm \epsilon_i\pm \epsilon_j\\
\pm\delta_k\pm\delta_l\\ \pm 2 \delta_k\end{array} & \pm \epsilon_i\pm
\delta_k\vspace{6pt}\\\osp{2M+1}{2N} &  \begin{array}{c}\pm \epsilon_i\pm
\epsilon_j\\ \pm \epsilon_i\\
\pm\delta_k\pm\delta_l\\ \pm 2 \delta_k\end{array} &
\begin{array}{c}\pm \epsilon_i\pm
\delta_k\\\pm \delta_k\end{array}
\end{array}
\eeq
\caption{The root systems of gl, sl and osp type superalgebras
in the standard basis $\epsilon_1,\dots,\epsilon_M, \delta_1,\dots,\delta_M$.
See for instance \cite{Frappat:1996pb} for more details.}
\end{table}

Consider the superalgebra $\gl{M}{N}$ first. Let $Q$ be a root
generator for the root $q=\epsilon_r-\delta_s$. The requirement
$\alpha\pm q \notin \Delta$ is satisfied for the following roots
\begin{equation}\label{eq:roots_anal_gl}
   \epsilon_i-\epsilon_j,\,
  \epsilon_i-\delta_k,\quad \mbox{ with } \quad
  \delta_k-\delta_l,\quad i,j\neq r,\,
  k,l\neq s\ \ .
\end{equation}
These give rise to the root system of a $\gl{M-1}{N-1}$
subalgebra. As a basis of the Cartan subalgebra one may choose the
Cartan generators $H_{\epsilon_i}$, $H_{\delta_k}$ which are
defined through eq.~\eqref{eq:useful_trick}. Evaluating
\begin{equation}\label{eq:cartan_algebra_reduction}
  \Ker (\epsilon_r-\delta_s)/\mathbb{C}H_{\epsilon_r-\delta_s}\ =\
 \Ker (\epsilon_r-\delta_s)\cap \Ker(\epsilon_r+\delta_s) \ =\
  \Ker \epsilon_r\cap \Ker \delta_s
\end{equation}
we deduce with the help of eq.~\eqref{eq:one_rank_case} that
\begin{equation}\label{eq:coh1}
  \Coh_{\epsilon_r-\delta_s}(\gl{M}{N})\ \simeq\  \gl{M-1}{N-1}  \ .
\end{equation}
The cohomological reduction of $\sgl{M}{N}$ is only slightly
different. As the roots of $\sgl{M}{N}$ and $\gl{M}{N}$ are the
same, the analysis~\eqref{eq:roots_anal_gl} remains unchanged. The
Cartan algebra of $\sgl{M}{N}$ can be viewed as the subalgebra of
the Cartan algebra of $\gl{M}{N}$ defined by $\Ker \str$, where we
have introduced the supertrace $\str:=\sum \epsilon_i -\sum
\delta_k$ . Therefore, eq.~\eqref{eq:cartan_algebra_reduction} has
to be replaced by
\begin{align*}
  \Ker \str\cap\Ker (\epsilon_r-\delta_s)/\mathbb{C}H_{\epsilon_r-\delta_s}
  &=\
  \Ker \str \cap\Ker (\epsilon_r-\delta_s)\cap \Ker
  (\epsilon_r+\delta_s)\\[2mm]
  {}&=\ \Ker(\sum_{i\neq r} \epsilon_i -\sum_{k\neq s} \delta_k)\cap
  \Ker \epsilon_r\cap \Ker \delta_s \ ,
\end{align*}
which leads to the Cartan subalgebra of $\sgl{M-1}{N-1}$.
Therefore we obtain
\begin{equation}\label{eq:coh2}
  \Coh_{\epsilon_r-\delta_s}(\sgl{M}{N})\ \simeq\  \sgl{M-1}{N-1}  \ .
\end{equation}
A similar analysis may be performed for osp type superalgebras. If
we choose $q=\epsilon_r \pm \delta_s$ then $\alpha\pm q$ is not a
root for all $\alpha$ from the following list
\begin{equation}\label{eq:red_d_roots}
  \pm \epsilon_i \pm \epsilon_j,\quad
  \pm \epsilon_i\pm\delta_k,\quad
  \pm\delta_k\pm \delta_l,\qquad
  i\neq j,\quad i,j\neq r,\quad k,l\neq s \ ,
\end{equation}
in the case of $\osp{2M}{2N}$ and
\begin{equation}\label{eq:red_b_roots}
  \pm \epsilon_i \pm \epsilon_j,\quad
  \pm \epsilon_i
  \pm \epsilon_i\pm\delta_k,\quad
  \pm\delta_k\pm \delta_l,\qquad
  i\neq j,\quad i,j\neq r,\quad k,l\neq s \ ,
\end{equation}
in the case of $\osp{2M+1}{2N}$. Those in
eq.~\eqref{eq:red_d_roots} correspond to the root system of an
$\osp{2M-2}{2N-2}$ subalgebra, while the roots in
eq.~\eqref{eq:red_b_roots} are associated with an
$\osp{2M-1}{2N-2}$ subalgebra. Again, one may take the  Cartan
Cartan generators $H_{\epsilon_i}$, $H_{\delta_k}$ as a basis of
the Cartan subalgebra. The cohomological reduction of the Cartan
subalgebra goes exactly as in
eq.~\eqref{eq:cartan_algebra_reduction}
\begin{equation}
  \Ker(\epsilon_r\pm\epsilon_r)/\mathbb{C}H_{\epsilon_r\pm\delta_s}
  \ = \  \Ker(\epsilon_r\pm\epsilon_r)\cap \Ker(\epsilon_r\mp\epsilon_r)
  \ = \ \Ker_{\epsilon_r}\cap\Ker_{\delta_s} \ .
\end{equation}
Therefore we conclude that
\begin{equation}\label{eq:coh3}
  \Coh_{\epsilon_r\pm \delta_s}(\osp{R}{2N})\ \simeq\  \osp{R-2}{2N-2}  \ ,
\end{equation}
for any choice of $R$. At this point we have determined the
cohomology $H_Q(\g)$ for all elements $Q$ that belong to the Cartan
eigenspace $\mathfrak{g}_q$ of an isotropic root $q$.
\medskip

From eqs.~(\ref{eq:coh1}, \ref{eq:coh2}, \ref{eq:coh3}) we may
infer that, up to isomorphism, the cohomological reduction of
$\mathfrak{g}$ with respect to $Q$ does not depend on the choice
of the isotropic root $q$. This gives rise to the following
question: How can we characterize $Q$s that give rise to different
Lie superalgebras $\Coh_Q(\mathfrak{g})$? In the following we want
to prove that the isomorphism class of the cohomological reduction
depends only on the rank of $Q$ in the fundamental representation.
To begin with we observe that an automorphism $\gamma$ of
$\mathfrak{g}$ induces an automorphism of the cohomology, i.e.\
\begin{equation}\label{eq:aut_coh}
  \Coh_Q (\mathfrak{g})\ \simeq\  \Coh_{\gamma(Q)} (\mathfrak{g}) \ .
\end{equation}
The main idea is to use the group of inner automorphisms provided
by the even subalgebra of $\mathfrak{g}$ in order to bring a general
$Q$ with vanishing self-bracket to some simpler form.

Consider the Lie superalgebra $\gl{M}{N}$ first. Let $V$, $V_M$
and $V_N$ be the fundamental $\gl{M}{N}$, $\bgl{M}$ and $\bgl{M}$
modules, respectively.
To bring $Q$ to some simpler form, we shall use the following
$\gl{M}{N}_{\bar{0}}\simeq \bgl{M}\oplus\bgl{N}$ module
isomorphism
\begin{equation}\label{eq:is_odd_1}
\gl{M}{N} _{\bar{1}}\ \simeq\  V_M \otimes_{\mathbb{C}}V^*_N\oplus
V_N\otimes_{\mathbb{C}}V^*_M  \ ,
\end{equation}
where $V^*$ denotes the dual representation. The module
isomorphism~\eqref{eq:is_odd_1} is provided by the invertible
linear map
\begin{align}
  \varphi(v\otimes \alpha)(a) &=\ v \alpha(a),&
  v\otimes\alpha &\in V_M\otimes_{\mathbb{C}}V^*_N\\[1mm] \notag
  \varphi(v\otimes \alpha)(u) &=\ 0,&  u&\in V_M\\[1mm] \notag
  \varphi(a\otimes \omega )(v) &=\  a \omega(v),&
  a\otimes\omega &\in V_N\otimes_{\mathbb{C}}V^*_M\\[1mm] \notag
  \varphi(a\otimes \omega )(b) &=\ 0,& b &\in V_N\ .
\end{align}
We say that $Q$ has rank $(k,l)$ if it can be represented as
\begin{equation}
  \varphi^{-1}(Q)\ = \ \sum_{i=1}^k v_i\otimes \alpha^i +\sum_{i=1}^l
  a_i\otimes \omega^i \ ,
\end{equation}
where all $v$'s, $a$'s, $\alpha$'s and $\omega$'s are linearly
independent among themselves. Clearly $k,l\leq \min(M,N)$. Let
$b_1,\dots,b_M$ denote a basis of $V_M$ and $f_1,\dots,f_N$ be a
basis of $V_N$. Denote by $b^i$, $f^k$ the dual bases. Then, from
the definition of the general linear group, there are elements
$A'\in\BGL{M}$, $B'\in\BGL{N}$ such that
\begin{equation}
  v_i \ =\  A'\cdot b_i,\quad \alpha^i \ =\  B'\cdot f^i,\qquad
  i=1,\dots,k
  \ .
\end{equation}
Moreover, the group elements $A',B'$ are not unique --- their
action on the remaining basis vectors $b_{k+1},\dots,b_M$ and
$f^{k+1},\dots,f^N$ is not fixed. Choosing an inner automorphism
$\gamma'=\Ad{A'}^{-1}\circ\Ad  B'^{-1}$ we see that one can bring
$Q$ to the simpler form
\begin{equation}\label{eq:first_red_gl}
  \varphi^{-1}(\gamma'(Q)) \ =\  \sum_{i=1}^k b_i\otimes f^i +\sum_{i=1}^l
  a'_i\otimes \omega'^i \ ,
\end{equation}
where $a'_i = B'^{-1}\cdot a_i$ and $\omega'^i = A'^{-1}\cdot\omega^i$.
The condition $Q^2=0$ is equivalent to the following constraints on
the vectors $a'_i,\, \omega'^i$ in eq.~\eqref{eq:first_red_gl}
\begin{equation*}
  f^j(a'_i)\ =\ 0,\quad \omega'^i(b_j)\ =\ 0\ ,
\end{equation*}
where $i=1,\dots,l$,  $j=1,\dots,k$. This means that the vectors
$a'_i$ lie entirely in the subspace of $V_N$ spanned by the basis
vectors $f_{k+1},\dots,f_N$, while the form $\omega'^i$ lies in
the subspace of $V^*_M$ that is spanned by the basis forms
$b^{k+1},\dots,b^M$. Therefore, the linear independence of $a'_i$,
$\omega'^i$ imposes an additional restriction on the rank $(k,l)$
of $Q$
\begin{equation}\label{eq:rank_cond_gl}
  k+l\ \leq\  \min(M,N)\ .
\end{equation}
The existence of the group elements $A''\in\BGL{M}$ and
$B''\in\BGL{N}$ satisfying
\begin{equation}
  A''\cdot b_i\ =\ b_i, \quad B''\cdot f^i \ = \ f^i\ ,
\end{equation}
for $i=1,\dots,k$ and
\begin{equation}
  a'_m\ =\ A''\cdot f_m,\quad \omega'^n\ = \ B''\cdot b^n\ ,
\end{equation}
for $m=k+1,\dots, k+l$ and $n=k+1,\dots,k+l$ is ensured by
eq.~\eqref{eq:rank_cond_gl}. Defining $\gamma''=\Ad A''^{-1}\circ
\Ad B''^{-1}$ we see that $Q$ can be brought into a standard form
which depends only on its rank $(k,l)$
\begin{equation}
  \varphi^{-1}((\gamma''\circ\gamma')(Q)) \ = \ \sum_{i=1}^k
  b_i\otimes f^i+\sum_{i=k+1}^{k+l} f_i\otimes b^i\ .
\end{equation}
We perform the cohomological reduction of $\gl{M}{N}$ with
respect to the fermionic generators
\begin{equation}\label{eq:canonical_form_gl}
  \varphi\left(\sum_{i=1}^k
  b_i\otimes f^i+\sum_{i=k+1}^{k+l} f_i\otimes b^i\right)
\end{equation}
by a lengthy but straightforward calculation. Thereby, we are lead
to the following statement
\begin{equation}
 \Coh_Q(\gl{M}{N}) \ \simeq\  \gl{M-\rank{(Q)}}{N-\rank{(Q)}} \ ,
\end{equation}
where the total rank of $Q$ is defined as $\rank{(Q)} = k+l\leq
\min{(M,N)}$.

The generalization to the superalgebras $\sgl{M}{N}$ is
straightforward. The procedure to bring $Q$ to the canonical
form~\eqref{eq:canonical_form_gl} is identical with the one
described in the $\gl{M}{N}$ case. The cohomological reduction of
$\sgl{M}{N}$ with respect to this canonical form of $Q$ may be
performed explicitly and leads to the expected result
\begin{equation}
 \Coh_Q(\sgl{M}{N})\ \simeq\
 \sgl{M-\rank{(Q)}}{N-\rank{(Q)}} \ .
 \end{equation}

Finally, let us also deal with the Lie superalgebras
$\osp{R}{2N}$, where $R=2M$ or $R=2M+1$. Denote by $V$, $V_R$ and
$V_{2N}$ the fundamental $\osp{R}{2N}$, $\bso{R}$ and $\bsp{2N}$
modules, respectively. Furthermore, let
$(\phantom{x},\phantom{x})$ be the symmetric invariant scalar
product in $V_R$ and $\langle \phantom{x},\phantom{x}\rangle$ be
the antisymmetric invariant scalar product in $V_{2N}$. For $R=2M$
we shall consider a basis $b_1,\dots,b_{2M}$ such that the matrix
elements of the scalar product $S_{ij}=(b_i,b_j)$ take the form
\begin{equation}\label{eq:sc_prod_v0_ev}
  S \ = \ \begin{pmatrix}
        0_{M\times M} & 1_{M\times M}\\[1mm]
    1_{M\times M} & 0_{M\times M}
      \end{pmatrix} \ ,
\end{equation}
while for $R=2M+1$ we shall consider a basis $b_1,\dots,b_{2M+1}$
such that the matrix elements of the scalar product $S_{ij}=(b_i,b_j)$
take the form
\begin{equation}\label{eq:sc_prod_v0_odd}
  S\ =\ \begin{pmatrix}
        0_{M\times M} & 1_{M\times M} & 0_{M\times 1}\\[1mm]
    1_{M\times M} & 0_{M\times M} & 0_{M\times 1}\\[1mm]
    0_{1\times M} & 0_{1\times M} & 1
      \end{pmatrix} \ .
\end{equation}
We also consider a basis $f_1,\dots,f_{2N}$ such that
the matrix elements of the scalar product $A_{ij}=
\langle f_i,f_j\rangle$ take the form
\begin{equation}\label{eq:sc_prod_v1}
  A \ = \ \begin{pmatrix}
        0_{N\times N} & -1_{N\times N}\\[1mm]
    1_{N\times N} & 0_{N\times N}
      \end{pmatrix} \ .
\end{equation}
With respect to the decomposition $V\simeq V_R\oplus V_{2N}$,
the invariant scalar product in $V$ is $G=S\oplus A$.

To bring $Q$ into some simpler form, we shall use the following
$\osp{R}{2N}_{\bar{0}}\simeq \bso{R}\oplus\bsp{2N}$ module
isomorphism
\begin{equation}\label{eq:is_odd_2}
\osp{R}{2N} _{\bar{1}}\ \simeq\  V_R \otimes_{\mathbb{C}}V_{2N}\ ,
\end{equation}
which is provided by the invertible linear map
\begin{align}\label{eq:isomorphism_odd_osp}
  \chi(s\otimes a)(b) & =\ s\langle a,b\rangle, \qquad
  s\otimes a \in V_R\otimes_{\mathbb{C}}V_{2N}\\[2mm] \notag
  \chi(s\otimes a)(t) & =\ a (s,t),\qquad
  t\in V_R,\; b\in V_{2N} \ .
\end{align}
We say that $Q$ has rank $k$ if it can be represented as
\begin{equation}
  \chi^{-1}(Q) \ = \ \sum_{i=1}^k s_i\otimes a_i\ ,
\end{equation}
where the $s$'s and $a$'s are linearly independent among
themselves. Of course $k\leq \min(R,2N)$. The condition $Q^2=0$
can be worked out from eqs.~\eqref{eq:isomorphism_odd_osp} to be
equivalent to the following constraints on the vectors $s_i$,
$a_i$
\begin{equation}
  (s_i,s_j)\ =\ 0,\qquad \langle a_i,a_j\rangle \ =\ 0\ ,
\end{equation}
for $i,j=1,\dots,k$. These conditions are compatible with the
linear independence of the $s_i$ and $a_i$ if and only if
\begin{equation}
  k\ \leq\  M,\qquad k\ \leq\  N\ .
\end{equation}
This restriction on the rank $k$ allows us to define some linearly
independent vectors $s_{k+1},.., s_R$ and $a_{k+1}, \dots, a_{2N}$
such that the matrix elements $(s_i,s_j)$, for $i,j=1,\dots,R$ and
$\langle a_i,a_j\rangle $, for $i,j=1,\dots,2N$ take the form in
eqs.~(\ref{eq:sc_prod_v0_ev}, \ref{eq:sc_prod_v0_odd}) and in
eq.~\eqref{eq:sc_prod_v1}, respectively. Therefore, from the
definition of the $\BSO{R}$ and $\BSP{2N}$ groups, there exist
elements $A\in\BSO{R}$ and $B\in\BSP{2N}$ such that
\begin{equation}
  s_i \ = \ A\cdot b_i,\qquad a_j \ =\  B\cdot f_j\ ,
\end{equation}
for $i=1,\dots,R$ and $j=1,\dots,2N$. We see that $Q$ can be
brought to a simple standard form depending only on its rank $k$
\begin{equation}
  \chi^{-1}(\gamma(Q))\  = \ \sum_{i=1}^k b_i \otimes f_i
\end{equation}
by acting with the inner automorphism $\gamma = \Ad A^{-1}\circ\Ad
B^{-1}$. We perform the cohomological reduction of $\osp{R}{2N}$
with respect to the fermionic generators
\begin{equation}
  \chi\left(\sum_{i=1}^k
  b_i\otimes f_i\right)
\end{equation}
by an explicit calculation. Thereby, we end up with the following
statement
\begin{equation}
 \Coh_Q(\osp{R}{2N})\  \simeq\ \osp{R-2 \rank{(Q)}}{2N-2 \rank{(Q)}} \ ,
\end{equation}
where $\rank{(Q)} = k\leq \min{([R/2],N)}$.


\subsection{Reduction of modules}
\label{sec:crm}

Let $\mathfrak{g}$ be one of the superalgebras considered in
sec.~\ref{sec:cr_lsa} and $Q$ be an odd element of $\mathfrak{g}$
with vanishing self-bracket $[Q,Q]=2Q^2=0$. As we have shown in
sec.~\ref{sec:cr_lsa}, there is a subalgebra
$\mathfrak{h}\subset\mathfrak{g}$ such that
$\Coh_Q(\mathfrak{g})\simeq \mathfrak{h}$.

First, notice that there is a $\mathfrak{h}$-stable filtration
\begin{equation}\label{eq:d_stab_filtr}
  V\ \supset\  \Ker_Q V \ \supset \ \Img_Q V \ .
\end{equation}
Indeed, $V$ is a $\mathfrak{h}$-submodule by restriction,
while $\Ker_Q V$ and $\Img_Q V$ are  $\mathfrak{h}$-submodules because
$\mathfrak{h}\subset\Ker_Q \mathfrak{g}$.
Finally, $\Ker_Q V\supset \Img_Q V$ follows from $Q^2=0$.

The existence of the $\mathfrak{h}$-stable filtration~\eqref{eq:d_stab_filtr}
means that $\Coh_Q(V)$ is generally a quotient of a submodule of the restriction
of $V$ to $\mathfrak{h}$.
However, if $V$ is self-dual, that is $V$ has an invariant non-degenerate scalar
product, then one can repeat the steps~(\ref{eq:start_mark} -- \ref{eq:ort_all},
\ref{eq:pr_drd_alg} -- \ref{eq:dir_sum_dec_coh_alg})  and prove a similar
$\mathfrak{h}$-module direct sum decomposition for $V$
\begin{equation}\label{eq:eq3_bla}
  V\ \simeq \ W \oplus E\oplus F\ ,
\end{equation}
where $W\simeq \Coh_Q(V)$ and $E=\Img_Q V$.
We list some of the properties of the subquotients $\Coh_Q(V)$
that will prove useful for the following.
\begin{lemma}\label{eq:3_prop}
Let $U,V$ be $\mathfrak{g}$ modules.
Then the following $\mathfrak{h}$-module isomorphisms hold
\begin{enumerate}
\item[a)] $\Coh_Q (U\oplus V) \ \simeq \ \Coh_Q(U) \oplus \Coh_Q(V)$
\item[b)] $\Coh_Q(V^*)\ \simeq\  \Coh_Q(V)^*$
\item[c)] $\Coh_Q(U\otimes V)\ \simeq \ \Coh_Q(U)\otimes\Coh_Q (V)$, if $U$, $V$ are finite dimensional.
\end{enumerate}
\end{lemma}

\begin{proof} a) The direct sum of the modules $U$ and $V$ means that
    there are orthogonal idempotents $e_U$ and $e_V$ such that
    they commute with the action of $\mathfrak{g}$ and $e_U U =
    U$, $e_V V =V$.
    One thus has
    \begin{align*}
      &e_U\Ker_Q (U\oplus V) \ =\ \Ker_Q(e_UU\oplus e_UV)\ =\ \Ker_Q
      U\\[1mm]
      &e_U\Img_Q (U\oplus V) \ =\ \Img_Q(e_UU\oplus e_UV)\ =\ \Img_Q U
    \end{align*}
and therefore $e_U \Coh_Q(U\oplus V)=\Coh_Q (U)$.
Similarly, $e_V \Coh_Q(U\oplus V)=\Coh_Q (V)$, which completes the proof of a).

\noindent b)
The elements of $\Coh_Q(V^*)$ are equivalence classes
$\pi(\mu)=\mu+Q\cdot V^*$ of forms $\mu\in \Ker_Q V^*$, that is
$\pi(\mu)$ is the equivalence class of forms that have the same
restriction on $\Ker_Q V$ as $\mu$. Therefore the projection map
$\pi$ is actually the restriction to $\Ker_Q V$. Moreover, the
condition that $\mu\in\Ker_Q V^*$ is equivalent to the requirement
that $\mu$ vanishes on $\Img_Q V$, that is $\Ker_Q V^*  \simeq
(V/\Img_Q V)^*$. These two observation lead to b)
\begin{align*}
  \Coh_Q (V^*) \ =\  \pi(\Ker_Q V^*)\  = \ \Ker_Q V^*\big\vert_{\Ker_Q
  V}&\simeq (V/\Img_Q V)^*\big\vert_{\Ker_Q V}\\[1mm]
   {}&=\ (\Ker_Q V/\Img_Q V)^*\ =\ \Coh_Q (V)^*
\end{align*}

\noindent c) There exist bases $h'_a$, $e'_i$, $f'_i$ of $U$ and
$h''_b$, $e''_j$, $f''_j$ of $V$ that bring the action of $Q$ to a
Jordan normal form
\begin{align*}
  Q\cdot h'_a & = \ 0,&
  Q\cdot e'_i & = \ 0,&
  Q\cdot f'_i & = \ e'_i\\[1mm]
  Q\cdot h''_b & =\ 0,&
  Q\cdot e''_j & =\ 0,&
  Q\cdot f''_j & =\ e'_j \ .
\end{align*}
Computing the action of $Q$ in the corresponding tensor basis of $U\otimes V$
we get that $\Ker_Q (U\otimes V)$ is spanned by
\begin{equation*}
  h'_a\otimes h''_b,\quad h'_a \otimes e''_j,\quad e'_i\otimes
  h''_b,\quad e'_i\otimes f''_j-(-1)^{|e'_i|}f'_i\otimes e''_j
\end{equation*}
and $\Img_Q(U\otimes V)$ is spanned by
\begin{equation*}
  h'_a \otimes e''_j,\quad e'_i\otimes
  h''_b,\quad e'_i\otimes f''_j-(-1)^{|e'_i|}f'_i\otimes e''_j \ ,
\end{equation*}
where $|\cdot|$ denotes the grading function.
Thus, $\Coh_Q(U\otimes V)$ is spanned by $h'_a\otimes h''_b$.
Finally we notice that $h'_a$ spans $\Coh_Q(U)$ and $h''_b$ spans $\Coh_Q(V)$,
which proves c).
\end{proof}

For a finite dimensional $\mathfrak{g}$-module $V$ we observe that
\begin{equation}
  \sdim \Coh_Q(V) \ =\  \sdim V \ .
\end{equation}
The statement follows from the existence of a Jordan normal form
for the representation of $Q$ in $V$. The vanishing of the
superdimension of a module $V$ is a necessary constraint for the
triviality of the cohomological reduction $\Coh_Q(V)$. Atypical
simple modules do not generally satisfy this constraint, while
projective modules do, see \cite{Germoni:thesis}.
\begin{lemma}\label{lem:proj}
  If $V$ is a finite dimensional projective $\mathfrak{g}$-module,
  then $\Coh_Q (V)=0$.
\end{lemma}
\begin{proof}
Let $\Gamma^+$ be the set of weights $\Lambda$ parametrizing the
simple finite dimensional $\mathfrak{g}$-modules $S(\Lambda)$.
Denote by $P(\Lambda)$ be the projective covers of $S(\Lambda)$,
that is the indecomposable $\mathfrak{g}$-modules with the top
$\topmod(P(\Lambda))=S(\Lambda)$. The projective module $V$ can
then be represented as
\begin{equation}\label{eq:ref_proof}
  V\ \simeq \ \bigoplus_{\Lambda\in \Gamma^+} d_\Lambda(V) P(\Lambda)\ ,
\end{equation}
where only a finite number of multiplicities $d_\Lambda(V)$ do not
vanish. Proving~\ref{lem:proj} becomes equivalent to proving that
$\Coh_Q(P(\Lambda))=0$ for any $\Lambda\in \Gamma^+$. We show in
the following that this task is equivalent to yet another one.

Define the induced modules
\begin{equation}\label{eq:def_ind_mod_pr}
 B(\Lambda)\ =\ \mathrm{Ind}^{\mathfrak{g}}_{\mathfrak{g}_{\bar{0}}}\mathrm{Res}_{
 \mathfrak{g}_{\bar{0}}}S(\Lambda)\  =\
 \mathcal{U}(\mathfrak{g})\otimes_{\mathfrak{g}_{\bar{0}}}
 S(\Lambda)
\end{equation}
which are finite dimensional and, according to \cite{Germoni2000:MR1840448}, are
also projective in the category of finite dimensional $\mathfrak{g}$-modules.
The surjective map $\Pi:B(\Lambda)\rightarrow S(\Lambda)$
\begin{equation}
  \Pi(u\otimes_{\mathfrak{g}_{\bar{0}}} s)\ =\ u\cdot s
\end{equation}
defines a projective $\mathfrak{g}$-module homomorphism.
By definition, $\topmod(B(\Lambda))$ is the direct sum of all
quotients of $B(\Lambda)$ by a maximal submodule.
Because $B(\Lambda)/\Ker \Pi = S(\Lambda)$ is simple,
$\Ker \Pi$ is a maximal submodule and
therefore $S(\Lambda)\subset \topmod (B(\Lambda))$.
On the other hand, decomposing $B(\Lambda)$ as in
eq.~\eqref{eq:ref_proof} we explicitly compute
\begin{equation}
  \topmod(B(\Lambda)) \ =\  \bigoplus_{\Lambda'\in \Gamma^+}
  d_{\Lambda'}(B(\Lambda)) \topmod(P(\Lambda'))\ =\
\bigoplus_{\Lambda'\in \Gamma^+} d_{\Lambda'}(B(\Lambda)) S(\Lambda') \ .
\end{equation}
which from $S(\Lambda)\subset \topmod(B(\Lambda))$ implies that
$P(\Lambda)$ must be a direct summand of $B(\Lambda)$. Thus, we
see that proving  $\Coh_Q(P(\Lambda))=0$ for any $\Lambda\in X^+$
is equivalent to proving that $\Coh_Q(B(\Lambda))=0$ for any
$\Lambda\in \Gamma^+$.

To compute $\Coh_Q(B(\Lambda))$ we construct a basis of
$B(\Lambda)$ which brings the action of $Q$ to a Jordan normal
form. Let $a_1,\dots,a_B$ be a basis of $\mathfrak{g}_{\bar{0}}$
and $b_1,\dots,b_F$ be a basis of $\mathfrak{g}_{\bar{1}}$.
According to Poincar\'e-Birkoff-Witt theorem, the elements of the
form
\begin{equation}\label{eq:basis_env_alg}
  b_{i_1} \cdots b_{i_k} a^{l_1}_1 \cdots a^{l_B}_B,\qquad
  k,l_i\geq 0,\; i_1< \cdots < i_k
\end{equation}
are a basis of $\mathcal{U}(\mathfrak{g})$. Given a basis
$s_\alpha$ of $S(\Lambda)$, the basis~\eqref{eq:basis_env_alg} of
$\mathcal{U}(\mathfrak{g})$ provides a basis
\begin{equation}\label{eq:basis_ind_module}
  b_{i_1} \cdots b_{i_k} \otimes s_\alpha, \qquad k\geq 0,\; i_1<
  \cdots < i_k\
\end{equation}
 of $B(\Lambda)$ by means of the def.~\eqref{eq:def_ind_mod_pr}.

Choosing a basis such that $b_1=Q$ immediately brings the action
of $Q$ to a Jordan normal form. It then becomes obvious that
$\Ker_Q (B(\Lambda))=\Img_Q(B(\Lambda))$ is spanned by the basis
vectors~\eqref{eq:basis_ind_module} with $i_1=1$.
\end{proof}

\subsection{Reduction of smooth functions on $G/G'$}
\label{sec:smooth_func} We shall restrict to Lie superalgebras
$\mathfrak{g}$ of the type considered in sec.~\ref{sec:cr_lsa}.
They all have an invariant, supersymmetric, consistent and
non-degenerate bilinear form
$(\phantom{x},\phantom{x}):\mathfrak{g}\times
\mathfrak{g}\rightarrow \mathbb{C}$. Consider a subalgebra
$\mathfrak{g}^{\prime}$ of $\mathfrak{g}$ such that
$(\phantom{x},\phantom{x})$ restricts to a  non-degenerate
bilinear form on $\mathfrak{g}'$ and suppose there is an odd
element $Q\in\mathfrak{g}^{\prime}$ with vanishing self-bracket
$[Q,Q]=2Q^2=0$.

According to eqs.~(\ref{eq:coh_isom},
\ref{eq:dir_sum_dec_coh_alg}), $\Coh_Q(\mathfrak{g})$ and
$\Coh_Q(\mathfrak{g}^{\prime})$ are isomorphic to some subalgebras
$\mathfrak{h}\subset \mathfrak{g}$ and, respectively,
$\mathfrak{h}^{\prime}\subset \mathfrak{g}^{\prime}$, with the
following direct sum decompositions
\begin{align}\label{eq:eq2}
\mathfrak{g}& \simeq\
\mathfrak{h}\oplus\mathfrak{e}\oplus\mathfrak{f}\\[1mm]
  \mathfrak{g}^{\prime}& \simeq\
\mathfrak{h}^{\prime}\oplus\mathfrak{e}^{\prime}\oplus\mathfrak{f}^{\prime}
\ \notag
\end{align}
as $\mathfrak{h}$  and $\mathfrak{h}^{\prime}$-modules,
respectively. Here $\mathfrak{e}=\Img_Q \mathfrak{g}$,
$\mathfrak{e}^{\prime}=\Img_Q \mathfrak{g}^{\prime}$. Our
assumption $Q\in\mathfrak{g}'\subset \mathfrak{g}$ implies the
following subalgebra inclusions
\begin{equation*}
 \mathfrak{h}\subset \mathfrak{h}',\qquad \mathfrak{e}\subset
 \mathfrak{e}',\qquad \mathfrak{f}\subset \mathfrak{f}'\ .
\end{equation*}

Let $\mathfrak{m}$ be the orthogonal complement of
$\mathfrak{g}^{\prime}$ in $\mathfrak{g}$ with respect to
$(\phantom{x},\phantom{x})$. The assumption on the non-degeneracy of the form
$(\phantom{x},\phantom{x})$ and of its restriction to $\mathfrak{g}^{\prime}$
implies the following facts on $\mathfrak{m}$:
\begin{enumerate}
\item[a)] $\mathfrak{m}$ is an $\mathfrak{g}^{\prime}$-module
\item[b)] $(\phantom{x},\phantom{x}) \vert_{\mathfrak{m}\times\mathfrak{m}}$
is an $\mathfrak{g}^{\prime}$-invariant non-degenerate scalar product
\item[c)] viewed as as an $\mathfrak{g}^{\prime}$-module by restriction,
$\mathfrak{g}$ decomposes as
\begin{equation}\label{eq:eq1}
  \mathfrak{g}\ \simeq \ \mathfrak{g}^{\prime}\oplus \mathfrak{m} \ .
\end{equation}
\end{enumerate}
Statements a) and b) are rather straightforward to prove, while c)
results from the construction of a projection on
$\mathfrak{g}^{\prime}$ with the inverted metric
$(\phantom{x},\phantom{x})
\vert_{\mathfrak{g}^{\prime}\times\mathfrak{g}^{\prime}}$, much
like in eq.~\eqref{eq:pr_drd_alg}. From eq.~\eqref{eq:eq3_bla} and
point c), $\mathfrak{m}$ decomposes as an
$\mathfrak{h}^{\prime}$-module into the direct sum
\begin{equation}\label{eq:eq3}
  \mathfrak{m}\ \simeq\  \mathfrak{n}\oplus \mathfrak{p}\oplus\mathfrak{q}\ ,
\end{equation}
where $\mathfrak{n}\simeq \Coh_Q(\mathfrak{m})$ and
$\mathfrak{p}=\Img_Q \mathfrak{m}$. Computing the cohomology of
the direct sum decomposition~\eqref{eq:eq1} with the help of
property a) of lemma~\ref{eq:3_prop} and eqs.~(\ref{eq:eq2},
\ref{eq:eq3}) we get an analogous decomposition
\begin{equation}\label{eq:eq1a}
 \mathfrak{h}\ \simeq\  \mathfrak{h}'\oplus \mathfrak{n}\ .
\end{equation}

One useful consequence of eqs.~(\ref{eq:eq1} -- \ref{eq:eq1a}) is
the following $\mathfrak{h}'$-module isomorphism
\begin{equation}\label{eq:eq_ru}
 \Coh_Q(\mathfrak{g}/\mathfrak{g}^{\prime})\ \simeq\  \Coh_Q(\mathfrak{m})\ \simeq
\ \mathfrak{n}\  \simeq\  \mathfrak{h}/\mathfrak{h}'\ .
\end{equation}

Let $\mathfrak{g}_{B,\bar{0}}$ be the Grassmann envelope of
$\mathfrak{g}$ with respect to some Grassmann algebra $B$ and
$\mathfrak{g}_{B,\bar{0},\dagger}$ a real form of the Lie algebra
$\mathfrak{g}_{B,\bar{0}}$ with respect to a complex anti-linear
involutive automorphism $\dagger$.
Suppose $G$ is a connected Lie supergroup with Lie algebra
$\mathfrak{g}_{B,\bar{0},\dagger}$ and $G^{\prime}$ is a connected subgroup of
$G$ with Lie algebra
$\mathfrak{g}^{\prime}_{B,\bar{0},\dagger}$.
Let $H$ denote the subgroup of $G$ with Lie
algebra $\mathfrak{h}_{B,\bar{0},\dagger}$ and $H^{\prime}$
the subgroup of $G^{\prime}$ with Lie algebra
$\mathfrak{h}'_{B,\bar{0},\dagger}$.
We want to perform the cohomological reduction of the space of smooth functions
$\mathcal{S}(G/G^{\prime})$ with respect to $Q$ and show that there is an $H$-module isomorphism
\begin{equation}\label{eq:cental_statement}
 \Coh_Q(\mathcal{S}(G/G^{\prime}))\simeq \mathcal{S}(H/H^{\prime})\ ,
\end{equation}
where $\mathcal{S}(H/H^{\prime})$ denotes the algebra of smooth
functions on $H/H^{\prime}$. Eq.~\eqref{eq:eq_ru} was already used
in sec.~\ref{sec:into_math} to give a local argument for the
isomorphism~\eqref{eq:cental_statement}. In order to prove the
claim~\eqref{eq:cental_statement}, we shall identify
$\mathcal{S}(G/G')$ with the space $\mathcal{S}(G)^{G'}$ of smooth
functions on $G$ invariant with respect to the right $G'$-action.
We perform the same identification for
$\mathcal{S}(H/H')=\mathcal{S}(H)^{H'}$.

Let us look closer at $\Img_Q \mathcal{S}(G/G^{\prime})$. The set
of points of $G/G'$ where all elements of $\Img_Q \mathcal{S}
(G/G^{\prime})$ vanish are precisely those points of $G/G'$ which
are invariant with respect to the action of $e^{\eta Q}$, where
$\eta$ is an odd Grassmann number. We denote this subset by
$(G/G')^Q$. Let $G^Q$ and $(G')^Q$ denote the subgroup of $G$ and,
respectively, $G'$ invariant with respect to the adjoint action of
$e^{\eta Q}$. These are the subgroups on which the vector field
$D(Q)$ corresponding to the adjoint action of $Q$ vanishes. This
means that $\Img_{D(Q)} \mathcal{S}(G)$ is the subset of smooth
functions on $G$ vanishing on $G^Q$.
\begin{lemma}
 The following equivalence of supermanifolds holds
\begin{equation}\label{eq:use_of_q_in_h}
(G/G')^Q \ =\  G^Q/(G')^Q
\end{equation}
\end{lemma}
\begin{proof}
In the neighborhood of $eG^{\prime}$, where $e$ is the identity of $G$, the
distinct equivalence classes of $G/G^{\prime}$ can be parametrized as
\begin{equation}\label{eq:geod_sc}
  e^v G^{\prime}\ ,
\end{equation}
where $v\in\mathfrak{m}_{B,\bar{0},\dagger}$ is small enough.
If we denote by $v$ the coordinate of the point~\eqref{eq:geod_sc} then we get
the geodesic system of coordinates at $eG^{\prime}$.
Indeed, the coordinate space $\mathfrak{m}_{B,\bar{0},\dagger}$ can be
identified with the tangent space at the point $eG^{\prime}$ with coordinates
$v=0$
\begin{equation*}
  (\mathcal{L}(v) f)(0) \ =\  \frac{d}{dt} (e^{t v}\cdot f)(0)
  \big\vert_{t=0}\ =\
  \frac{d}{dt}f(-t v)\big\vert_{t=0}
  \ =\  -(v(f))(0)\ ,
\end{equation*}
where $v\in\mathfrak{m}_{B,\bar{0},\dagger}$ and $\mathcal{L}$ denotes the Lie derivative.

The exponential mapping
\begin{equation}\label{eq:exp_map}
 v\rightarrow e^v G^{\prime}
\end{equation}
can be extended to the whole tangent space
$\mathfrak{m}_{B,\bar{0},\dagger}$.
This extension is in general no longer injective, that is it
ceases to be a system of coordinates. However, assuming Hopf-Rinow
theorem can be generalized to supermanifolds \cite{hopfrinow}, the
map~\eqref{eq:exp_map} must be surjective , that is any group
element $g\in G$ can be represented in the form
\begin{equation*}
 g \ =\  e^v g^{\prime}
\end{equation*}
for some $v\in\mathfrak{m}_{B,\bar{0},\dagger}$ and $g^{\prime}\in G^{\prime}$.
Using this global representation, one can easily see that
$(G/G^{\prime})^Q$ is the image of exponential mapping~\eqref{eq:exp_map}
restricted to $\Ker_Q \mathfrak{m}_{B,\bar{0},\dagger}$.
If follows that $G^Q$  has a transitive action on $(G/G^{\prime})^Q$.
Its stabilizer at $eG^{\prime}\in (G/G^{\prime})^Q$ with respect to the left
action on $G^Q$ is $(G^{\prime})^Q=G^Q \cap G^{\prime}$.
This completes the proof of claim~\eqref{eq:use_of_q_in_h}.
\end{proof}

\begin{cor} Let $L(Q)$ denote the vector field corresponding to
the left action of $Q$. Then one has
\begin{equation}\label{eq:img_eq}
 \Img_{Q} \mathcal{S}(G/G') \ =\  \Img_{L(Q)} \mathcal{S}(G)^{G'} \ = \
 \Img_{D(Q)} \mathcal{S}(G)^{G'}\ =\ (\Img_{D(Q)} \mathcal{S}(G))^{G'}
 \end{equation}
\end{cor}
\begin{proof}
The first equality results from the identification
$\mathcal{S}(G/G')=\mathcal{S}(G)^{G'}$ while the second equality
is a consequence of $Q\in\mathfrak{g}'$. To prove the last
equality notice that $\Img_{D(Q)} \mathcal{S}(G)$ is composed of
functions on $G$ vanishing on  $G^Q$. Then $(\Img_{D(Q)}
\mathcal{S}(G))^{G'}$ becomes the space of functions on $G/G'$
vanishing on the submanifold $G^Q/G'$. Notice that
$G^Q/G'= G^Q/(G')^Q$, because both supermanifolds are
$G^Q$-transitive and have the same stabilizer $(G')^Q = G^Q\cap
G'$. Therefore, according to eq.~\eqref{eq:use_of_q_in_h},
$(\Img_{D(Q)} \mathcal{S}(G))^{G'}$ can be seen as the space of
functions on $G/G'$ that vanish on $(G/G')^Q$. This, however,
coincides with the definition of $\Img_Q\mathcal{S}(G/G')$.
\end{proof}

We have analogous obvious equalities for the kernel of $Q$
\begin{equation}\label{eq:ker_eq}
  \Ker_{Q} \mathcal{S}(G/G') \ = \ \Ker_{L(Q)} \mathcal{S}(G)^{G'} \ = \
  \Ker_{D(Q)} \mathcal{S}(G)^{G'}\ = \ (\Ker_{D(Q)} \mathcal{S}(G))^{G'}\ .
\end{equation}
Combining eqs.~(\ref{eq:img_eq}, \ref{eq:ker_eq}) we get the following
prescription for computing the cohomology
\begin{equation}
 \Coh_Q(\mathcal{S}(G/G')) \ =\ (\Coh_{D(Q)}(\mathcal{S}(G)))^{G'}\ .
\end{equation}
Let us now concentrate on computing $\Coh_{D(Q)}(\mathcal{S}(G))$.

The image of a function $f$ under the projection map
$\pi:\mathcal{S}(G)\rightarrow \mathcal{S}(G)/\Img_{D(Q)}
\mathcal{S}(G)$ given by
\begin{equation}\label{eq:proj_map_diff_1}
 \pi(f)\ =\ f+\Img_{D(Q)} \mathcal{S}(G)
\end{equation}
is the equivalence class of functions which have the same restriction on
$G^Q$ as $f$,
that is
\begin{equation}\label{eq:proj_map_diff_2}
 \pi(f) \ =\  f\vert_{G^Q}\ .
\end{equation}
In particular any function whose restriction to $G^Q$ vanishes must
be in the image of $Q$.

Notice that the left and the right $G$-actions on $\mathcal{S}(G)$
induce corresponding left and right $G^Q$-action on the quotient
space $\mathcal{S}(G)/\Img_{D(Q)} \mathcal{S}(G)$
\begin{equation*}
 L(X) \pi (f) \ :=\  \pi (L(X)f),\qquad R(X) \pi (f) \ : =\  \pi(R(X)f) \ .
\end{equation*}
\begin{lemma}
The following isomorphism of $H$-modules and commutative algebras holds
\begin{equation}\label{eq:bla_not_3}
 \Coh_{D(Q)}(\mathcal{S}(G))^{G'}\ \simeq\  \mathcal{S}(H/H')\ .
\end{equation}
\end{lemma}
\begin{proof}
 If $X\in\mathfrak{g}$ and $f\in\Ker_{L(Q)}\mathcal{S}(G)^{G'}$, then
$$L([Q,X])\pi (f)\ =\  \pi(L[Q,X] f)\ =\  \pi(L(Q)L(X)f)\ =\  \pi(D(Q)L(X)f)\ =\
0\ ,$$
because the left and right $\mathfrak{g}$ actions on $\mathcal{S}(G)$ commute and $\Ker \pi = \Img_{D(Q)}\mathcal{S}(G)$. This shows
that the space of functions $\Coh_{D(Q)}(\mathcal{S}(G))^{G'}$ is
left invariant with respect to the action of $\mathfrak{e}$.
Denote by $N$ the subgroup of $G$ with the Lie superalgebra
$\mathfrak{e}$. The latter being an ideal of $\Ker_Q
\mathfrak{g}$, $N$ is a normal subgroup of $G^Q$ with $H =
N\backslash  G^Q$. Then eq.~\eqref{eq:bla_not_3} claims that
$\Coh_{D(Q)}(\mathcal{S}(G))^{G'}$ is a space of functions on
$N\backslash G^Q/G' = H/G'=H/H'$. The last equality comes from the
fact that both $H/G'$ and $H/H'$ are $H$-transitive and have the
same stabilizer $H'= G'\cap H$.
\end{proof}

In conclusion we wee that the cohomology of a smooth function on
$G/G'$ is computed by restricting it to $H/H'\subset G/G'$. We
denote this restriction map by $\rho$.

\subsection{Reduction of smooth tensor forms on $G/G'$}
\label{sec:red_smooth_forms}
Let $T_k(G/G')$ be the space of smooth tensor forms of rank $k$
on $G/G'$.
We claim that eq.~\eqref{eq:cental_statement} can be generalized to
\begin{equation}\label{eq:tensors_coh}
 \Coh_Q(T_k(G/G'))\ \simeq\  T_k(H/H')\ ,
\end{equation}
where $T_k(H/H')$ is the space of smooth tensor forms of rank $k$ on $H/H'$.
We shall only give a local argument.
Introducing the geodesic coordinates~\eqref{eq:geod_sc}, one can perform
the following identification in the neighborhood of the point $eG'\in G/G'$
\begin{equation*}
 T_k(G/G')\ \simeq \ \mathcal{S}(G/G')\otimes \mathfrak{m}^{\otimes k}\ .
\end{equation*}
This local trivialization extends to an isomorphism of $G'$-modules.
Using the property c) of lemma~\ref{eq:3_prop}, we get
\begin{equation}\label{eq:make_one_ref}
 \Coh_Q(T_k(G/G'))\ \simeq \ \Coh_Q(\mathcal{S}(G/G'))\otimes
 \Coh_Q(\mathfrak{m})^{\otimes k}\simeq \mathcal{S}(H/H')\otimes
 \mathfrak{n}^{\otimes k}\simeq T_k(H/H')\ .
\end{equation}
Most probably, one can give a global argument for the
claim~\eqref{eq:tensors_coh} by introducing the frame bundle
\begin{equation*}
 T_k(G/G')\simeq (\mathcal{S}(G)\otimes F(G)^{\otimes k})^{G'}\ ,
\end{equation*}
where $F(G)$ is the moving frame attached to every point of $G$, which
is built out of the components of the Maurer-Cartan form.

In conclusion, the cohomology of a tensor form on $G/G'$ is
computed, as can be seen from eq.~\eqref{eq:make_one_ref}, by restricting it i) to the submanifold $H/H'$ and ii) to
the tensor space of $H/H'$. Step ii) is equivalent to throwing out
all components of the tensor not lying in the tensor space of
$H/H'$ seen as a submanifold of $G/G'$. We denote this restriction
map by $\rho$ again.


\subsection{Reduction of $L_2(G/G^{\prime})$}
\label{sec:crl2}

We want to refine~\eqref{eq:cental_statement}
and  show that the elements of $\Coh_Q(L_2(G/G^{\prime}))$ are square integrable
with respect to some $H$-invariant measure on $H/H^{\prime}$, that is
\begin{equation}\label{eq:blabla_i_am_tired}
 \Coh_Q(L_2(G/G^{\prime}))\ \simeq\  L_2(H/H^{\prime})\ .
\end{equation}
In order to do so, introduce the geodesic coordinates $v$ of
eq.~\eqref{eq:geod_sc}. Let $v_{\mathfrak{n}}$, $v_{\mathfrak{p}}$
and $v_{\mathfrak{q}}$ denote the projection of $v$ onto the real
Grassmann envelope of the direct summand $\mathfrak{n}$,
$\mathfrak{p}$ and, respectively, $\mathfrak{q}$ in
eq.~\eqref{eq:eq3}. We then embed $\mathcal{S}(H/H^{\prime})$ into
$\mathcal{S}(G/G^{\prime})$ by means of the injection map
\begin{equation}\label{eq:inj_map_diff}
 i(f)(v) \ = \ f(v_{\mathfrak{n}})e^{\alpha (v_{\mathfrak{p}},v_{\mathfrak{q}})}\ ,
\end{equation}
where $v$ is small enough and $\alpha$ is, for the moment, an
arbitrary number. Notice that eq.~\eqref{eq:inj_map_diff} defines
the function $i(f)$ globally. Indeed, the
definition~\eqref{eq:inj_map_diff} allows to compute the action of
the enveloping Lie superalgebra $\mathcal{U}(\mathfrak{g})$ on
$i(f)$. The latter can be extended to the action of the group $G$,
whose knowledge is enough to define the values of $i(f)$ at any
point of $G/G^{\prime}$.

The most important property of the injection map~\eqref{eq:inj_map_diff} is
\begin{equation}\label{eq:important_property}
 \pi\circ i \ = \ \id\ ,
\end{equation}
where $\pi$ is the projection of eqs.~(\ref{eq:proj_map_diff_1}, \ref{eq:proj_map_diff_2}).
As a consequence, any element of $\Ker_Q \mathcal{S}(G/G^{\prime})$ can be
represented in the form
\begin{equation}\label{eq:fundamental_eq_diff}
 i(f)+\mathcal{L}(Q) h \ ,
\end{equation}
where $\mathcal{L}(Q)$ denotes the Lie derivative with respect to $Q$.

We now prove~\eqref{eq:blabla_i_am_tired} by showing that for a proper
choice of $\alpha$ in eq.~\eqref{eq:inj_map_diff} one has
\begin{equation}\label{eq:fin_res_1}
 \langle i(f_1),i(f_2)\rangle_{G/G^{\prime}} \ = \ \langle f_1,
f_2\rangle_{H/H^{\prime}} \ .
\end{equation}
The equation should be understood as follows: i) the existence of one side implies the existence of the other side and ii) for a
$G$-invariant scalar product
$\langle\phantom{x},\phantom{x}\rangle_{G/G^{\prime}}$ on
$L_2(G/G^{\prime})$ induced by the $G$-invariant measure on
$G/G^{\prime}$ there is a corresponding $H$-invariant scalar
product $\langle\phantom{x},\phantom{x}\rangle_{H/H^{\prime}}$ on
$L_2(H/H^{\prime})$ induced by the $H$-invariant measure on
$H/H^{\prime}$.

Indeed, let the measure on $G/G^{\prime}$ be given locally by
$d\mu_G(v)=w(v)dv$. Suppose $i(f)$ is $L_2$ normalizable. Then its norm can be
written in the form
\begin{equation*}
 \int_{H/H^{\prime}} d\mu_H f^2\ ,
\end{equation*}
where $d\mu_H$ is a measure on $H/H^{\prime}$ locally defined by a weight
function $w'(v_{\mathfrak{n}})$ obtained by integrating
\begin{equation*}
 w(v)e^{2\alpha(v_{\mathfrak{p}},v_{\mathfrak{q}})}\ ,
\end{equation*}
with respect to the coordinates $v_{\mathfrak{p}}$ and $v_{\mathfrak{q}}$.
Notice that there is always a choice of $\alpha$ such that $w'$ exists even for
non-compact homogeneous spaces $G/G^{\prime}$.
Of course, in order to perform the integration yielding the explicit form of
$w'$ one must work with an atlas of $G/G^{\prime}$.
However, the only thing that matters to us is its $H$-invariance or,
equivalently, the $H$-invariance of the scalar product $\langle \phantom{x},
\phantom{x}\rangle_{H/H'}$ associated to it by eq.~\eqref{eq:fin_res_1}.
We thus check%
\begin{equation}\label{eq:h-invariance}
 \langle i(\mathcal{L}(X) f_1),i(f_2)\rangle_{G/G^{\prime}} + \langle i(
f_1),i(\mathcal{L}(X)f_2)\rangle_{G/G^{\prime}}\ =\ 0,\qquad
X\in\mathfrak{h}\ .
\end{equation}
Notice that $(v_{\mathfrak{p}},v_{\mathfrak{q}})$ is $Q$-exact
because its restriction to $v_{\mathfrak{q}}=0$ vanishes.
Therefore $\mathcal{L}(X)(v_{\mathfrak{p}},v_{\mathfrak{q}})$ is
also $Q$-exact, because $[Q,X]=0$. Finally,
\begin{equation}\label{eq:hopefully_the_last_one}
 (\mathcal{L}(X) i(f))(v) - i(\mathcal{L}(X) f)(v) \ = \
 f(v_{\mathfrak{n}})\alpha e^{\alpha(v_{\mathfrak{p}},v_{\mathfrak{q}})}
 \mathcal{L}(X)(v_{\mathfrak{p}},v_{\mathfrak{q}})
\end{equation}
is $Q$-exact as well, because
$f(v_\mathfrak{n})e^{\alpha(v_{\mathfrak{p}}, v_{\mathfrak{q}})}$
is $Q$-invariant. We then use the exactness of the
expression~\eqref{eq:hopefully_the_last_one} to commute the Lie
derivative $\mathcal{L}(X)$ with the injection $i$ in
eq.~\eqref{eq:h-invariance}.

We conclude this section by noticing that eq.~\eqref{eq:fin_res_1}
can be written in an equivalent way as
\begin{equation}\label{eq:localization}
 \langle f_1,f_2\rangle_{G/G^{\prime}}\  = \ \langle
\rho(f_1),\rho(f_2)\rangle_{H/H^{\prime}}\ ,
\end{equation}
where  $f_1,f_2\in\Ker_QL_2(G/G^{\prime})$.
This is the localization phenomenon.

\section{Cohomological reduction in the field theory}
\label{reductionfield}

We are now prepared to revisit the sigma models on $G/G'$. We have
shown in sec.~\ref{sec:observables} how the local observables of
the sigma model on $G/G'$ can be constructed from functions on
$L_2(G/G')$ and (some well behaved subspace of the space of
smooth) tensor forms on $G/G'$. The results of
sec.~(\ref{sec:smooth_func}--\ref{sec:crl2}) straightforwardly
imply  that the cohomological reduction of the space of local
observables in the sigma model on $G/G'$ coincides precisely with
the space of local observables in the sigma model on $H/H'$, that
is
\begin{equation}
 \Coh_Q( \mathcal{F}_{G/G'})\ \simeq \  \mathcal{F}_{H/H'}
\end{equation}

We now look at
correlation functions of local fields $\cal{O}$ that are $Q$-invariant. As the
results of the previous section suggest, we shall demonstrate that
any correlation function of such fields can be computed in the
$H/H'$ coset superspace theory.

First we need to compute the cohomological reduction of the action
$S_{G/G'}$ associated to the Lagrangian in
eq.~\eqref{eq:lag_origin}. Since the Lagrangian is entirely fixed
by a $G$-invariant metric and an exact $G$-invariant 2-form, we
can apply the results of sec.~\ref{sec:red_smooth_forms} in order
to compute their cohomology class. The classes of the two tensor
forms are computed by restricting them to the points of the
submanifold $H/H'$ and to its tensor space respectively. As a
result we obviously get an $H$-invariant metric and an exact
$H$-invariant 2-form on $H/H'$. Employing the restriction map
$\rho$ of sec.~(\ref{sec:smooth_func},\ref{sec:red_smooth_forms}),
we conclude that
\begin{equation}\label{eq:red_act}
 \rho(S_{G/G'})\ =\ S_{H/H'}
\end{equation}
is an action for the sigma model on $H/H'$ with a similar kinetic
term and $B$-field structure as $S_{G/G'}$.
The pullback of eq.~\eqref{eq:red_act} takes a more familiar form to usual cohomological calculations in field theory
\begin{equation*}
 S_{G/G'} = S_{H/H'} + \mathcal{L}(Q) R\ ,
\end{equation*}
where $\mathcal{L}(Q)$ denotes the Lie derivative with respect to $Q$ and $R$ is some residual functional, obviously non $G$-invariant.
The possibility of constructing $G$-invariant terms $\mathcal{L}(Q) R$ out of non $G$-invariant terms $R$ is a special feature of the supergroup symmetry.
According to one of the main ideas behind cohomological reduction, the $Q$-exact term in the action does not contribute to the calculation of correlation functions of $Q$-invariant local fields.

To make things more precize, notice that the localization formula~\eqref{eq:localization}
for the computation of the scalar product of $Q$-invariant
functions can be generalized to the integral of any $Q$-invariant
object. Therefore, we trivially obtain from
eq.~\eqref{eq:def_corr_func}
\beqa\label{eq:final_corr}
\vac{\prod_{i=1}^N\calO_i(x_i)}_{G/G^{\prime}}&=&
\int_{ \mathcal {\mathcal{H} } } \, \text{d}\mu_H e^{-\rho(S_{G/G'})}
\prod_{i=1}^N\rho\left(\calO_i\right)(x_i)\nonumber\\[2mm] &=&\vac { \prod_ { i=1 }
^N\rho\left(\calO_i\right)(x_i)}_{H/H^{\prime}}\ . \eeqa
where we have used eq.~\eqref{eq:red_act}. Consequently, the
subsector of the sigma model on $G/G'$ which we obtain through
cohomological reduction is composed of the localized observables
$\rho(\mathcal{O}_i)$. Finally, using the central
statement~\eqref{eq:final_corr}, we conclude that this subsector
is exactly identified with the local observables of the sigma
model on $H/H'$.

\section{Applications}
\label{reductionexamples}

In the first subsection we discuss applications of cohomological
reduction to conformal field theory. In the second subsection we
present a general treatment of sigma models on supercoset spaces
$G/G^{\mathbb{Z}_2}$ defined by a degree two automorphism, that is
on symmetric superspaces. The last subsection deals with some
specific examples involving automorphisms of degree four.

\subsection{Conformal field theory}
\label{sec:cft}
The cohomological reduction we have described in the previous two
subsections allows us to identify certain simple subsectors of the
parent theory in which all correlation functions can be computed
explicitly through the reduced model. The latter is often much
simpler than the original theory. In fact, we shall find many
examples below in which the subsector is a free or even
topological field theory. The existence of such simple subsectors
may signal very special features of the parent model. In
particular, it can imply its scale invariance.

In order to make a more precise statement we need a bit of
preparation. Let us recall that the coset $G/G'$ gives rise to a
family of sigma models which is parametrized by the metric $\tG$
and the $B$-field $\tB$. $G$-invariance of the action determines the
two background fields up to a finite number of parameters. Upon
quantization, these parameters may be renormalized. This
renormalization of $\tG$ and $\tB$ can affect the properties of
our theory and in particular of its stress tensor.

Let us now consider the quantized $G/G'$ model that comes with
some fixed choice of $\tG$ and $\tB$. The associated stress tensor
$T_G$ is conserved and symmetric. On the other hand, the trace of
$T_G$ may be non-zero due to quantum effects. The components of
$T_G$ are $G$-invariant, i.e.\ they commute with all generators $X
\in {\mathfrak{g}}$. In general, $T_G$ can be decomposed into a sum
$T_G = \sum_i T^{(i)}_G$ of terms where each of the summands
$T^{(i)}_G$ belongs to a single indecomposable representation of
${\mathfrak{g}}$. We say that $T_G$ is a true $G$-invariant if every summand $T_G^{(i)}$ is a direct summand. This must be distinguished from more generic cases for
which some of the summands $T^{(i)}_G$, although transforming in
the trivial representation of $\mathfrak{g}$, are coupled to other
fields through the action of a nilpotent symmetry generator $N$
from the center of the enveloping Lie superalgebra
$\mathcal{U}(\mathfrak{g})$. In this case, $T^{(i)}_G = N
t^{(i)}_G$ for some field $t^{(i)}_G$, which is called a
logarithmic partner of $T^{(i)}_G$.

Let us now assume that the tensor $T_G$ is a true $G$-invariant in
the sense we have described above. Suppose furthermore that the
theory contains a \emph{conformal}  subsector $H/H'$ with a
non-vanishing stress tensor $T_H$. According to our assumption,
$T_H$ is conserved, symmetric and traceless. Consequently, the
stress tensor of the original theory must be conserved, symmetric
and traceless up to some $Q$-exact terms. Since we assumed $T_G$
to be a true invariant, though, non of its components --- and in
particular the trace of $T_G$ --- can be obtained by acting with
an element of $\mathcal{U}(\mathfrak{g})$ on some other fields.
Hence, $T_G$ must be traceless and hence the $G/G'$ model is
conformal.

Let us stress again that our assumption on $T_G$ to be a true
invariant is rather strong. We are not prepared to state precise
conditions under which this assumption is actually satisfied in
general. However, when the superspaces $G/G'$ have at most one
degree of freedom in the choice of $\mathsf{G}$ and $\mathsf{B}$
one can get a simple constraint for the conformality of the parent
theory from the conformality of the cohomological subsector
theory: the sigma model $G/G'$ is conformal if $H/H'$ is conformal
and its central charge is non-zero. Indeed, in this case
$\mathsf{G}$ and $\mathsf{B}$ is either  proportional to i) a
single $\mathfrak{g}$ true invariant or to ii) a single invariant
socle of a $\mathfrak{g}$-indecomposable module. If $H/H'$ is the
conformally invariant maximal cohomological reduction with a
non-zero central charge, then $T_G$ cannot be an invariant socle.
Otherwise we would get a contradiction, because its 2-point
function would vanish and the 2-point function of $T_G$ must
coincide with the 2-point function of $T_H$. The latter, however,
cannot vanish because the central charge of the conformal $H/H'$
sigma model is non-zero.


\subsection{Sigma models on symmetric superspaces}
\begin{figure}[t]
\centerline{\xymatrix{ &\ar@{.>}[d] &\ar@{.>}[d] & \\ \ar@{.>}[r]
& \frac{\gl{M+m}{N+n}}{\gl{M}{N}\oplus \gl{m}{n}} \ar[r] \ar[d]&
\frac{\gl{M+m-1}{N+n-1}}{\gl{M-1}{N-1}\oplus \gl{m}{n}} \ar@{.>}[r] \ar[d]& \\
\ar@{.>}[r]&  \frac{\gl{M+m-1}{N+n-1}}{\gl{M}{N}\oplus
\gl{m-1}{n-1}} \ar[r] \ar@{.>}[d]&
\frac{\gl{M+m-2}{N+n-2}}{\gl{M-1}{N-1}\oplus \gl{m-1}{n-1}}
\ar@{.>}[r] \ar@{.>}[d]& \\ & & & }} \caption{\small{Possible
cohomological reductions of
$\gl{M+m}{N+n}/\gl{M}{N}\oplus\gl{m}{n} $.} }
\end{figure}
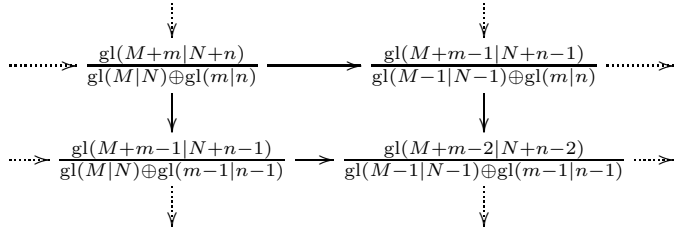

In this section, we want to present a classification of the
cohomological reductions of $\mathbb{Z}_2$ cosets, i.e. of
symmetric superspaces. These supermanifolds $G/G^{\prime}$ have
the property that $G^{\prime}$ is a direct product of supergroups
of which at most two are simple. For each simple factor whose
superalgebra contains nilpotent fermionic elements, we can perform
the cohomological reduction. Reductions performed with $Q$
operators that come from different simple factors commute with
each other. As an example, consider the coset space
$\mathfrak{g}/\mathfrak{g}^{\prime}=\gl{M+m}{N+n}/\gl{M}{N}\oplus
\gl{m}{n} $. The denominator has two simple factors, so that we
can reduce in two ways as outlined in figure $1$.

In table $2$ below, we describe the different cohomological sectors of
all possible sigma models on symmetric superspaces. We only write
down the complex case, but different reality conditions can then
easily be taken into consideration.

\begin{table}[h!]
\label{classification} \centerline{
\begin{tabular}{cc p{4cm}}
$\mathcal{R}$ & $\mathcal{M}$ &$\ \,$ Comments\\
\hline & & \\
$\frac{\psl{1}\oplus \psl{1}}{\psl{1}}$ &
$\frac{\psl{1+\alpha}\oplus \psl{1+\alpha}}{\psl{1+\alpha}}$
&$\ \,$ C  \vspace{5pt}\\
$\frac{\gl{1}{1}\oplus \gl{1}{1}}{\gl{1}{1}}$ &
$\frac{\gl{1+\alpha}{1+\alpha}\oplus
\gl{1+\alpha}{1+\alpha}}{\gl{1+\alpha}{1+\alpha}}$
&$\ \,$ T \vspace{5pt} \\
$\frac{\bsgl{R}\oplus \bsgl{R}}{\bsgl{R}}$ &
$\frac{\sgl{R+\alpha}{\alpha}
\oplus \sgl{R+\alpha}{\alpha}}{\sgl{R+\alpha}{\alpha}}$ & \vspace{5pt} \\
$\frac{\bgl{R+S}}{\bgl{R}\oplus \bgl{S}}$ &
$\frac{\gl{R+S+\alpha+\beta}{\alpha+\beta}}{\gl{R+\alpha}{
\alpha}\oplus
\gl{S+\beta}{\beta}}$ & $\ \,$ T for $R=0$ or $S=0$\vspace{5pt} \\
$\frac{\gl{R}{S}}{\bgl{R}\oplus \bgl{S}}$ &
$\frac{\gl{R+\alpha+\beta}{S+\alpha+\beta}}{\gl{R+\alpha}{\alpha}\oplus
\gl{\beta}{S+\beta}}$& $\begin{array}{l} \text{ C for
$R=S=1$}\\\text{ T for
$R=0$ or $S=0$}\end{array}$\vspace{5pt}\\
$\frac{\psl{1}\oplus\psl{1}}{\psl{1}}$ &
$\frac{\psl{2\alpha}}{\osp{2\alpha}{2\alpha}}$ & $\ \,$ C
\vspace{5pt}\\
$\frac{\gl{2}{2}}{\osp{2}{2}}$ &
$\frac{\gl{2+2\alpha}{2+2\alpha}}{\osp{2+2\alpha}{2+2\alpha}}$ &
$\ \,$ T
\vspace{5pt}\\
$\frac{\bsgl{R}}{\bso{R}}$ &
$\frac{\sgl{R+2\alpha}{2\alpha}}{\osp{R+2\alpha}{2\alpha}}$ & $\
\,$
T for $R=1$ \vspace{5pt}\\
$\frac{\bsgl{2R}}{\bsp{2R}}$ &
$\frac{\sgl{2\alpha}{2R+2\alpha}}{\osp{2\alpha}{2R+2\alpha}}$ & \vspace{5pt}\\
$\frac{\sgl{1}{2R}}{\osp{1}{2R}}$ &
$\frac{\sgl{1+2\alpha}{2R+2\alpha}}{\osp{1+2\alpha}{2R+2\alpha}}$&
\vspace{5pt}\\
$\frac{\bso{R}\oplus \bso{R}}{\bso{R}}$ &
$\frac{\osp{R+2\alpha}{2\alpha}\oplus
\osp{R+2\alpha}{2\alpha}}{\osp{R+2\alpha}{2\alpha}}$ &
$\begin{array}{l} \text{ C for $R=2$} \\ \text{ T for $R=0$ or
$R=1$}\end{array}$\vspace{5pt}\\
$\frac{\bsp{2R}\oplus \bsp{2R}}{\bsp{2R}}$ &
$\frac{\osp{2\alpha}{2R+2\alpha}\oplus
\osp{2\alpha}{2R+2\alpha}}{\osp{2\alpha}{2R+2\alpha}}$ & $\ \,$ T
for $R=0$ \vspace{5pt}\\
$\frac{\osp{1}{2R}\oplus \osp{1}{2R}}{\osp{1}{2R}}$ &
$\frac{\osp{1+2\alpha}{2R+2\alpha}\oplus
\osp{1+2\alpha}{2R+2\alpha}}{\osp{1+2\alpha}{2R+2\alpha}}$ & \vspace{5pt}\\
$\frac{\bso{R+S}}{\bso{R}\oplus \bso{S}}$ &
$\frac{\osp{R+S+2\alpha+2\beta}{2\alpha+2\beta}}{\osp{R+2\alpha}{2\alpha}\oplus
\osp{S+2\beta}{2\beta}}$ & $\begin{array}{l} \text{ C for
$R=S=1$}\\ \text{ T
for $R=0$ or $S=0$}\end{array}$ \vspace{5pt}\\
$\frac{\osp{R}{2S}}{\bso{R}\oplus \bsp{2S}}$ &
$\frac{\osp{R+2\alpha+2\beta}{2S+2\alpha+2\beta}}{\osp{R+2\alpha}{2\alpha}\oplus
\osp{2\beta}{2S+2\beta}}$ & $\ \,$ T for $R=0$ or $S=0$ \vspace{5pt}\\
$\frac{\bsp{2R+2S}}{\bsp{2R}\oplus \bsp{2S}}$ &
$\frac{\osp{2\alpha+2\beta}{2R+2S+2\alpha+2\beta}}{\osp{2\alpha}{2R+2\alpha}
\oplus \osp{2\beta}{2S+2\beta}}$ & $\ \,$ T for $R=0$ or
$S=0$\vspace{5pt}\\
$\frac{\osp{2}{2R+2S}}{\osp{1}{2R}\oplus \osp{1}{2S}}$ &
$\frac{\osp{2+2\alpha+2\beta}{2R+2S+2\alpha+2\beta}}{\osp{1+2\alpha}{2R+2\alpha}
\oplus \osp{1+2\beta}{2S+2\beta}}$ & \vspace{5pt}\\
 $\frac{\bso{2R}}{\bgl{R}}$&
$\frac{\osp{2R+2\alpha}{2\alpha}}{\gl{R+\alpha}{\alpha}}$ & $\ \,$
T
for $R=0,1$\vspace{5pt}\\
 $\frac{\bsp{2R}}{\bgl{R}}$& $\frac{\osp{2\alpha}{2R+2\alpha}}{\gl{\alpha}{
R+\alpha}}$ & $\ \,$  T for $R=0$ \vspace{5pt}\\
\hline
\end{tabular}
} \caption{The left column presents the possible minimal
non-trivial sectors labelled by $R$, $S$ and the right one the
chain of models to which they belong. We denote by T the models
that have a topological subsector and by C those models that are
conformally invariant.}
\end{table}

Some of the minimal subsectors are topological. This occurs when
the whole Lagrangian is in the image of $Q$, which is the case
whenever the right side of table ($2$) can be brought to
the form $\mathfrak{g}/\mathfrak{g}$. This happens for the $\GL{N}{N}$,
$\OSP{2N+1}{2N}$ and $\OSP{2N}{2N}$  principal chiral models as
well as for the cosets
\beq \frac{\GL{N+p}{N+q}}{\GL{N}{N}\times \GL{p}{q}} \quad
\frac{\GL{2N}{2N}}{\OSP{2N}{2N}}\quad
\frac{\OSP{2N+p}{2N+2q}}{\OSP{2N}{2N}\times \OSP{p}{2q}}\quad
\frac{\OSP{2N}{2N}}{\GL{N}{N}}\nonumber\ . \eeq

On the other hand, some cohomological reductions lead to free
conformal field theories, for which there are only two
possibilities. Either they reduce to the $c=1$ free boson model to
the $c=-2$ theory of a pair of symplectic fermions. The former
case occurs for the $\OSP{2N+2}{2N}$ principal chiral model and
the real Grassmannians
\beq \frac{ \OSP{2+2m+2n}{2m+2n} }{ \OSP{1+2m}{2m}\times
\OSP{1+2n}{2n} } \ ,\label{eq:series_1} \eeq
whereas the latter occurs for the $ \PSL{N} $ principal chiral
model as well as for the cosets
\beq \frac{\GL{m+n+1}{m+n+1}}{\GL{m+1}{m}\times \GL{n}{n+1}}\qquad
\frac{\PSL{2N}}{\OSP{2N}{2N}}\ .\label{eq:series_2} \eeq
As was shown in \cite{Candu:thesis} by direct computation of the
all loop $\beta$ function, these are the only sigma models on
symmetric spaces that are conformally invariant. The superspaces
$G/G'$ in eqs~(\ref{eq:series_1}, \ref{eq:series_2}) have only one
radius and no $G$-invariant $B$-field. We thus see that the
argument of sec.~\ref{sec:cft}  leads to the same classification
of conformally invariant sigma models, while this time being
non-perturbative in nature.

\subsection{Examples involving generalized symmetric spaces}

We will now turn our attention to a few generalized symmetric
spaces in which the denominator supergroup $G'$ is left invariant
under the action of some automorphism $\Omega:G \mapsto G$ of
order four. We are not attempting to provide a classification of
such cosets, but restrict our discussion to three interesting
examples. The first series of models contains theories whose
minimal subsector is given by the sigma model for $AdS_5\times
S^5$ and $AdS_2\times S^2$ spaces. The second and third example
extend the construction of superspheres and complex projective
spaces, respectively. In all three families of models we shall
identify previously unknown candidates for conformal cosets, see
eqs.\ \eqref{eq:ex1CFT}, \eqref{eq:ex2CFT} and \eqref{eq:ex3CFT}.
\bigskip

\noindent {\sc Example 1:} We look at the coset \beq
\mathfrak{g}/\mathfrak{g}^{\prime}=\frac{\ensuremath{\text{psu}\left({2(M+m)}|{
2(N+n)}
\right)}} {\osp{2m}{2n}\oplus \osp{2N}{2M}} \eeq defined for
$M+m=N+n$ by the following automorphism of order four:
$\Omega=-st\circ \text{Ad}_{X}\circ\text{Ad}_{Y}$ with \beq
X=\left(\begin{array}{cc|cc}& \mathbbm{1}_{M+m}& &
\\-\mathbbm{1}_{M+m}& & & \\\hline & & & \mathbbm{1}_{N+n}\\& &
-\mathbbm{1}_{N+n} & \end{array}\right)\quad
Y=\text{diag}\left(\mathbbm{1}_m,-\mathbbm{1}_{2M+m},
\mathbbm{1}_{N+2n}, -\mathbbm{1}_N \right)\ . \eeq Here, in order
to properly define the automorphism, one has to embed the superalgebra
$\text{psu}\left({2(M+m)}|{2(N+n)} \right)$ in the fundamental
representation of $\text{su}\left({2(M+m)}|{2(N+n)} \right)$. The
invariant subalgebra $\mathfrak{g}^{\prime}$ is a direct sum for which
the grading of the second summand is opposite that of the first
one. In order to know the number of free parameters in the metric
and $B$ field defining the model, we have to know how the $\Omega$
eigenspaces transform under the action of $\mathfrak{g}^{\prime}$. The
result is
\beq
\mathfrak{m}_{1}\cong
{\small\Yvcentermath1{\yng(1)\otimes \yng(1)}}\qquad\mathfrak { m }
_2\cong \left(\varnothing\otimes
{\small\Yvcentermath1{\yng(1,1)}}\right)\oplus \left(
{\small\Yvcentermath1{\yng(1,1)}}\otimes \varnothing\right)\qquad
\mathfrak{m}_{3}\cong {\small\Yvcentermath1{\yng(1)\otimes\yng(1)}}\ .
\eeq Here, as well as in the following examples, $\varnothing$
denotes the trivial representation,
$\small\Yvcentermath1{\yng(1)}$ the fundamental representation and
$\small\Yvcentermath1{\yng(1)^*}$ its dual. Tensor products of the
fundamental representation and of its dual that possess certain
permutation symmetry are denoted by the appropriate Young
tableaux.

We want to mention three special cases for these cosets
\begin{itemize}
\item
Without loss of generality, we choose $Q$ to lie only in the
second direct summand of $\mathfrak{g}^{\prime}$. Assuming that $M=N$
and thus $m=n$, we see that the maximal reduction in this case
leads to the sigma model on the $\mathbb{Z}_2$ coset
$\PSU{2m}/\OSP{2m}{2m}$, which is conformal. We thus arrive at the
conclusion that the sigma models on the $\mathbb{Z}_4$ coset
spaces
\beq \label{eq:ex1CFT} \mathcal{C}_{(N,n)} \ \cong \
\frac{\PSU{2(N+n)}}{\OSP{2n}{2n}\times \OSP{2N}{2N}} \eeq
are promising candidates for conformal sigma models for all non
negative values of $N$ and $n$.
\item
If we specialize to $M=n=2$, $m=N=0$ and change the reality
conditions appropriately, we obtain the well known $\mathbb{Z}_4$
coset space $\text{PSU}\left(2,2|4\right)/\BSO{4,1}\times\BSO{5}$
whose bosonic base is $AdS_5\times S^5$.  This model cannot be
reduced any further, since $\mathfrak{g}^{\prime}$ is purely bosonic.
It constitutes the maximal reduction of the two parameter discrete
family of models 
\beq \mathcal{M}_{(m,n)}=\frac{\text{PSU}\left(2m+2n+2,
2|2m+2n+4\right)}{\OSP{2m}{ 2m+2, 2} \times\OSP {2n}{2n+4}}\ .
\eeq
\item
Setting $M=n=0$, $m=N=1$ and again taking the appropriate boundary
conditions, leads to the space
$\text{PSU}\left(1,1|2\right)/\BSO{2}\times\BSO{2}$ whose bosonic
base is $AdS_2\times S^2$. This case is the maximal reduction of
the family of sigma models with $\mathfrak{g}=\psu{2(m+n+1)}$ and
$\mathfrak{g}^{\prime}=\osp{2m+2}{2m}\oplus\osp{2n+2}{2n}$, subject to
a certain
reality conditions. %
\end{itemize}

\bigskip

\noindent {\sc Example 2:} We are interested in the $\mathbb{Z}_4$
coset \beq
\mathfrak{g}/\mathfrak{g}^{\prime}=\frac{\osp{M+2m}{2N+2n}}{\osp{p}{2q}
\oplus\osp{M-p}{
2(N-q)}\oplus \un{m}{n}}\ .  \eeq The corresponding automorphism is
$\Omega=Ad_{X}$ with \beqa X=\left(\begin{array}{cc|ccc} I_{M}^p &
& &&  \\ & J_{2m}& & &\\\hline & & I_{N}^q &&\\& & & I_N^q\\& & &&
J_{2n}\end{array}\right)\text{ where } \begin{array}{l}
I^p_n=\left(\begin{array}{cc} \mathbbm{1}_p &0\\0&
-\mathbbm{1}_{n-p}\end{array}\right)\\
\vspace{5pt}\\J_{2n}=\left(\begin{array} {cc} 0 & \mathbbm{1}_n\\
-\mathbbm{1}_n & 0\end{array}\right) \ .\end{array} \eeqa Under
the action of $\mathfrak{g}^{\prime}$, the $\Omega$ eigenspaces
transform as \beqa \mathfrak{m}_{1}&\cong&
\big({\small\Yvcentermath1{\yng(1)\otimes\varnothing\otimes
\yng(1)}}\big)\oplus
\big({\small\Yvcentermath1{\varnothing\otimes\yng(1)\otimes
\yng(1)}}\big)\nonumber\\\mathfrak{m}_2&\cong&
\left(\varnothing\otimes\varnothing\otimes
{\small\Yvcentermath1{\yng(1,1)}}\right)\oplus
\left(\varnothing\otimes \varnothing\otimes
{\small\Yvcentermath1{\yng(1,1)}^*}\right)\oplus
\big(\small\Yvcentermath1{\yng(1)\otimes \yng(1)\otimes
\varnothing}\big)\nonumber\\ \mathfrak{m}_3&\cong&
\big({\small\Yvcentermath1{\yng(1)\otimes\varnothing\otimes
\yng(1)^*}}\big)\oplus
\big({\small\Yvcentermath1{\varnothing\otimes\yng(1) \otimes
\yng(1)^*}}\big)\ , \eeqa where by $U\otimes V\otimes W$ we
understand a module defined as the tensor product of the $U,V,W$
representations of respectively $\osp{p}{2q}$, $\osp{M-p}{2(N-q)}$
and $\un{m}{n}$. When selecting the fermionic operator $Q\in
\mathfrak{g}^{\prime}$, we choose it to be fully contained in one of
the direct summands of $\mathfrak{g}^{\prime}$ . Since the first
two lead, after suitable choice of the parameters $M,N,p,q$, to
the same reduction, we will assume, that $Q$ is either in
$\osp{M-p}{2(N-q)}$ or in $\un{m}{n}$.

\begin{itemize}
\item
If now we have $M=2N$ and $p=2q$, then we can pursue the reduction
of the first type until we get rid of the orthosymplectic parts in
$\mathfrak{g}^{\prime}$ to arrive at the sigma model on the symmetric
space $\OSP{2m}{2n}/\U{m}{n}$ which is not a conformal theory.
\item If on the other hand $m=n$, then taking
the second type of reduction can be used to remove the unitary
part of $\mathfrak{g}^{\prime}$ so as to obtain the sigma model on the
symmetric space $\OSP{M}{2N}/\OSP{p}{2q}\times\OSP{M-p}{2(N-q)}$,
which is a conformal field theory for $p=1$, $q=0$ and $M=2N+2$.
We therefore come to the conclusion that for all $N,n\in
\mathbb{N}$ the sigma models on the homogeneous spaces
\beq \label{eq:ex2CFT} \mathcal{C}_{(N,n)} \ \cong \
\frac{\OSP{2N+2+2n}{2N+2n}}{\OSP{2N+1}{2N}\times \U{n}{n}} \eeq
are candidates for conformally invariant sigma models. For $n=0$
they reduce to the symmetric spaces $S^{2N+1|2N}$, i.e. the
superspheres,
whereas for $N=0$ they remain a $\mathbb{Z}_4$ homogeneous space.%
\end{itemize}
\bigskip

\noindent  {\sc Example 3:} The last case under consideration is
the $\mathbb{Z}_4$ coset \beq
\mathfrak{g}/\mathfrak{g}^{\prime}=\frac{\un{M+2m}{N+2n}}{\un{p}{q}\oplus
\un{M-p}{N-q}\oplus \un{m}{n}\oplus \un{m}{n}} \eeq defined by the
automorphism $\Omega=Ad_Y$, where \beq
Y=\left(\begin{array}{cc|cc} I_{M}^p & & &  \\ & J_{2m}& &
\\\hline & &  I_{N}^q &\\ & & & J_{2n}\end{array}\right)\ . \eeq
We need not spell out the decomposition of $\mathfrak{m}_i$ in modules
of $\mathfrak{g}^{\prime}$, it suffices to say that the only
representations that appear are of the kind $A\otimes B\otimes
C\otimes D$, where $A,B,C,D$ are either the trivial, fundamental
or dual fundamental of respectively $\un{p}{q}$, $\un{M-p}{N-q}$,
the first $\un{m}{n}$ and the second $\un{m}{n}$. We choose $Q$ to
be diagonally embedded in the $\un{m}{n}\oplus \un{m}{n}$ part of
$\mathfrak{g}^{\prime}$, so that the reduction procedure sends the
parameters $m$ and $n$ to $m-1$ and $n-1$. If $m=n$, then the
reduction terminates with the symmetric space
$\U{M}{N}/\U{p}{q}\times \U{M-p}{N-q}$. The sigma models with this
target spaces are conformal for $M=N$ and $p=q\pm1$, with the
special case $p=1$ and $q=0$ corresponds to the complex symmetric
superspaces $\CP{N}$. In conclusion, we can state that the sigma
models on the homogeneous spaces
\beq \label{eq:ex3CFT} \mathcal{C}_{(M,N,n)} \ \cong \
\frac{\U{M+N+2n}{M+N+2n}}{\U{M+1}{M}\times \U{N-1}{N}\times
\U{n}{n}\times \U{n}{n}}\ , \eeq
are expected to be conformal for values of $M,N,n\in \mathbb{N}$
with $N>0$.

\subsection{Extensions of the cohomological reduction}

In this section, we want to expand the technique of cohomological
reduction to encompass Wess-Zumino-Witten and Gross-Neveu models.

The Wess-Zumino term on the supergroup $G$ with the superalgebra
$\mathfrak{g}$ takes the form \beqa \calS_{WZ}&=&-\frac{i}{24
\pi}\int_{B} \left(g^{-1}\text{d} g,
\left[g^{-1}\text{d}g,g^{-1}\text{d} g\right]\right)\nonumber\\
&=&-\frac{i}{24 \pi}\int_{B} d^3x\, \epsilon^{\alpha\beta\gamma}
f_{abc}\,J_{\alpha}^a J_{\beta}^b  J_{\gamma}^c \ ,\eeqa where
$\epsilon^{\alpha\beta\gamma}$ is the antisymmetric tensor and
$f_{abc}$ are the structure constants. Thus, the Wess-Zumino term
is a trilinear combination of the left invariant currents $J$ and
can be cohomologically reduced in a similar fashion as the 
bilinear $\tG$ and $\tB$ terms.

A straight forward if lengthy computation shows that, choosing a
fermionic operator $Q\in \mathfrak{g}$ that squares to zero, the
cohomologically reduced model is the Wess-Zumino-Witten model on
the supergroup $H$ whose superalgebra is $\Coh_Q(\mathfrak{g})$. Thus,
the Wess-Zumino-Witten models on a supergroup reduce in the same
way as the principal chiral models on the same supergroup. It can
even happen that a principal chiral model and a WZW model both
reduce to the same theory. An interesting example of that is
furnished by the $\PSU{N}$ WZW and principal chiral, the maximal
reduction of which is the $c=-2$ free theory of a single pair of
symplectic fermions.
\smallskip

Let us now turn to a very different theory, namely the osp$(m|2n)$
Gross Neveu model for $m$ free real fermions and $n$ pairs of
bosonic ghosts. The free part of the action is determined through
{\beq \calS^{\text{GN}}_{\text{free}} \ = \ \frac{1}{2\pi}
\int_{\Sigma} d^2z \biggl[
 \sum_{i=1}^{m} \bigl(\psi_i \bar\partial \psi_i +
  \bar{\psi}_i \partial \bar{\psi}_i\bigr)
  + \sum_{a=1}^{n} \bigl(\beta_a \bar \partial \gamma_a + \bar{\beta}_a
\partial
  \bar{\gamma}_a\bigr) \biggr] \ . 
\eeq} This action defines a conformal field theory with central
charge $c= \frac{m-2n}{2}$ and both left and right
$\widehat{\mathrm{osp}}(m|2n)$ current symmetry at level $k=1$. The
conformal dimension of the fundamental fields $(\psi_i,
\beta_a,\gamma_b)$ is $h=\frac{1}{2}$. The interaction term for
this theory is {\begin{eqnarray*}
\calS^{\text{GN}}_{\text{int}}=\frac{g^2}{2\pi} \int_{\Sigma} d^2z
\left[\sum_{i=1}^{m} \psi_i \bar{\psi}_i + \sum_{a=1}^n (\gamma_a
\bar{\beta}_a - \beta_a \bar{\gamma}_a)\right]^2\ .
\end{eqnarray*}}
An alternative way of understanding this model, is to think of the
free part of the theory as a free field representation of the
Wess-Zumino-Witten model of $\OSP{2m}{2n}$ at level one and of the
interacting model as being a current-current perturbation thereof.

If at the free point, we take the fermionic nilpotent generator
\beq Q=\frac{1}{4\pi i }\left\{\oint_0
dz\left(\psi_1+i\psi_2\right)\beta_1 -\oint_0 d\bar{z}
\left(\bar{\psi}_1+i\bar{\psi}_2\right)\bar{\beta}_1\right\}\ ,
\eeq we see that \beq \calL_{\text{int}}^{GN}=
\left[\sum_{i=3}^{m} \psi_i \bar{\psi}_i + \sum_{a=2}^n (\gamma_a
\bar{\beta}_a - \beta_a \bar{\gamma}_a)\right]^2+Q\cdot B \ .\eeq
It is furthermore not hard to see that a field is in the
cohomology of $Q$ if and only if it does not contain any
contribution from the fields $\psi_1,\psi_2, \beta_1,\gamma_1$.
The cohomologically reduced model is therefore the $osp(m-2|2(n-1))$
Gross Neveu model. An interesting example is obtained if we set
$m=2n+2$, in which case the maximal reduction of the Gross Neveu
model is the massless Thirring model of two real fermions, which
is dual to the theory of a compactified free boson. On the other
hand, the compactified free boson provides the endpoint of the
cohomological reduction of the sigma models on the superspheres
$S^{2n+1|2n}$. In \cite{Candu:2008yw ,Mitev:2008yt}, it was
proposed that there exists a duality between the osp$(2n+2|2n)$ Gross
Neveu models and the sigma models on $S^{2n+1|2n}$ and what we see
here supports this claim.


\section{Conclusion and Outlook}

In this paper we have studied correlation functions of quantum
field theories with internal supersymmetry. Most of the analysis
was tailored towards 2-dimensional sigma models on coset
superspaces $G/G'$. For such theories we chose a BRST operator $Q
\in \mathfrak{g'}$ and performed a cohomological reductions to another
coset sigma model $H/H'$. Fields of the latter were shown to be in
one-to-one correspondence with the cohomology
$\Coh_Q(\cal{H}^{G/G'})$ in the state space $\cal{H}^{G/G'}$ of
the $G/G'$ model. This correspondence preserves all correlators.
Let us stress once more that reductions of this type are certainly
not restricted to sigma models. Similar arguments also apply to
non-geometric theories such as e.g.\ Gross-Neveu or
Landau-Ginsburg models.

Before we conclude, let us sketch a number of possible
applications. The results of this work have been used already in
\cite{Candu:2009ep} for an investigation of boundary spectra in
sigma models on complex projective superspaces \CP{N}. For
vanishing coupling, it is an easy combinatorial exercise to
determine the spectra of these sigma models. Our strategy then was
to calculate the spectrum at finite coupling by assuming so-called
Casimir evolution of conformal weights
\cite{Quella:2007sg,Candu:2008yw}. This assumption was tested
carefully both through background field expansions and extensive
numerical studies. From the spectrum at vanishing coupling along
with the assumed Casimir evolution, the boundary partition
functions can be constructed up to two unknown functions. This is
were the cohomological reduction comes in. In fact, we then showed
that the two unknown functions may be identified as conformal
weights of fields belonging to the symplectic fermion subsector of
the \CP{N} coset model. Since symplectic fermions are free, we
could fix all the remaining freedom and find an exact analytic
expression for boundary partition functions of the \CP{N} sigma
model.

Another application concerns the issue of conformal symmetry. For
sigma models on symmetric superspaces $G/G^{{\mathbb{Z}_2}}$ we
found a free subsector $H/H^{{\mathbb{Z}_2}}$ if and only if the
original model was conformal. We also provided several examples of
more general coset superspaces $G/G'$ that possess a free
subsector. Though we are not prepared to argue that the
corresponding $G/G'$ coset sigma models are in fact conformal, we
believe that this is the case, at least for an appropriate choice
of the background fields $\tG$ and $\tB$. In any case, the issue
certainly deserves further investigation.

As we have stated in the introduction, one of the main motivations
for the study of superspace sigma models comes from the AdS/CFT
correspondence. It is likely that the ideas of this work can be
adjusted so as to apply to models that are relevant for the study
of strings in AdS geometries. In the case of $AdS_3$, for example,
correlation functions of chiral primaries have been computed in
the NSR formalism using the explicit solution of the WZNW model on
the bosonic space $H_3^+ \times$ SU(2). A closer look at the
results of \cite{Gaberdiel:2007vu,Pakman:2007hn} shows that most of the
intricate features of the full WZNW model cancel out from the
correlation function of chiral primaries. The answer looks much
simpler than one might naively expect, very much like a three
point function in some free field theory. We hope to re-derive and
extend these findings through a cohomological reduction, after
re-phrasing the computation in the target space supersymmetric
hybrid formalism \cite{Berkovits:1999im}. A detailed analysis in
currently being pursued.

Concerning the study of strings on $AdS_5$, concrete applications
seem a little more speculative. Within the pure spinor approach,
strings in $AdS_5 \times S^5$ may be described  by coupling some
superspace coset model $G/G^{\mathbb{Z}_4}$ to the pure spinor
ghost sector \cite{Berkovits:2000fe}. The coset model is one of
the examples we described in section 5. Since its denominator
group is purely bosonic, we cannot apply our ideas to this matter
sector alone, without taking the ghost sector into account. The
combined model does possess PSL(N$|$N) supersymmetry, but the
supercharges are only realized as on-shell symmetries. In order to
apply the ideas described in this paper one would need to pass to
a new formulation in which some of the target space supercharges
become off-shell symmetries. Symmetries of the light-cone gauge
fixed Green-Schwarz superstring on $AdS_5\times S^5$ (see e.g.\
\cite{Arutyunov:2009ga}), one the other hand, are described by two
copies of a centrally extended psu(2$|$2) algebra which share
their central elements. It might be feasible to use some of the
corresponding fermionic generators for a cohomological reduction.
It seems interesting to explore such cohomological reductions for
string theory in $AdS_5$ backgrounds.


\bibliographystyle{JHEP}
\bibliography{Redbib}

\end{document}